\documentclass{article}
\usepackage{epsfig}
\usepackage{amsbsy,latexsym}
\usepackage{amsmath}
\usepackage{amssymb, mathrsfs}
\usepackage[mathscr]{eucal}
\usepackage{hyperref}
\usepackage{amsfonts}
\usepackage{amsthm}
\usepackage{amsmath}
\usepackage{amsfonts}
\usepackage{latexsym}
\usepackage{amssymb}
\usepackage{amscd}
\usepackage[latin1]{inputenc}
\usepackage{verbatim}

\usepackage[english]{babel}

\newtheorem{definition}{Definition}[section]
\newtheorem{theorem}[definition]{Theorem}
\newtheorem{lemma}[definition]{Lemma}
\newtheorem{corollary}[definition]{Corollary}
\newtheorem{proposition}[definition]{Proposition}
\theoremstyle{definition}
\newtheorem{remark}[definition]{Remark}

\newtheorem{example}[definition]{Example}

\newcommand\tr{ \operatorname{Tr} } 

\newcommand{\Ep}{\mathcal{E}}
\newcommand{\Md}{\mathcal{M}_{d}}

\newcommand{\Bh}{\mathcal{B}(\mathcal{H})}
\newcommand{\Hi}{\mathcal{H}}
\newcommand{\Me}{\mathcal{M}_{\mathcal{E}}}
\newcommand{\Mi}{\mathcal{M}_{\mathcal{E}^{\infty}}}
\newcommand{\Ne}{\mathcal{N}_{\mathcal{E}}}
\title{ Multiplicative Properties of Quantum Channels 
}

\author{Mizanur Rahaman}

\begin{document}
\maketitle
{\em
Department of Mathematics and Statistics,
University of Regina}\\
{\em Regina, Saskatchewan S4S 0A2, Canada} 

%%%%%%%%%%%%%%%%%%% title and abstract %%%%%%%%%%%%%%%%%%%%%%%%%%%%%%

\begin{abstract}
In this paper, we study the multiplicative behaviour of quantum channels, mathematically described by trace preserving, completely positive maps on matrix algebras. It turns out that the multiplicative domain of a unital quantum channel has a close connection to its spectral properties. A structure theorem (Theorem \ref{main theroem}), which reveals the automorphic property  of an arbitrary unital quantum channel on a subalgebra,  is presented. Various classes of quantum channels (irreducible, primitive etc.) are then analysed in terms of this stabilising subalgebra. The notion of the \emph{multiplicative index} of a unital quantum channel is introduced, which measures the number of times a unital channel needs to be composed with itself for the multiplicative algebra to stabilise.
We show that the maps that have trivial multiplicative domains are dense in completely bounded  norm topology in the set of all unital completely positive maps. Some applications in quantum information theory are discussed.  
\end{abstract}
%%%%%%%%%%%%%%%%%%%%%%%
 
 %%%%%%%%%%%%%%%%%%%%%%%%%%%%%%
\section*{Introduction}
Quantum channels are the most general input-output transformations allowed by quantum mechanics. Physically, they play a central role in quantum information theory, where they represent the communication from a sender to a receiver (\cite{Watrous-book}), in quantum information processing (\cite{Werner}, \cite{Perinotti}, \cite{Chiribella}), and in the theory of quantum open systems (see the monograph \cite{math-language-qit}).
 The steps of a quantum computation and also the effects of errors and noise on quantum registers are modeled as quantum channels.  
In quantum statistical mechanics involving finitely many particles,  
the typical domain on which a quantum channel acts is the $d\times d$ complex matrices which we denote by $\Md$.

Although the maps in discussion are linear, their domain, the $d\times d$ complex matrices $\Md$, has an algebraic structure. Thus, it is of interest to investigate the multiplicative nature of such linear maps. The domain on which the quantum channel is multiplicative is known as the \emph{multiplicative domain} and the main theme of this paper is to study this domain in detail. The scheme of \emph{quantum error correction}, one of the major themes of current research in information theory, was successfully analysed from the point of view of the multiplicative nature of channels ( see \cite{johnston}, \cite{kribs-spekkens}, \cite{c-j-k}). Also, the  multiplicative domain appears to be a useful area
to explore in the study of private algebras and complementary quantum channels (see \cite{levick}).   

From a purely operator algebraic perspective, the multiplicative domains of positive and completely positive maps have been studied by many authors for independent interests (\cite{choi1}, \cite{stormer2007}, \cite{bunce-salinas}). In this context, it is essential to mention the work of {Bulinski\u\i} in {\cite{Bulinskii}} where the author 
considers the dynamical system $(\mathcal{M}, \omega, \tau)$ and investigates the asymptotic automorphic behaviour of the dynamical system, where $\mathcal{M}$ is a von Neumann algebra, $\tau= (\tau_{t})_{t\geq 0}$ is a family of unital normal  completely positive maps on $\mathcal{M}$ parametrised over the positive real numbers and $\omega$ is a faithful normal state on $\mathcal{M}$ such that $\omega(x)=\omega\circ\tau_{t}(x)$, for every $t\geq 0$ and $x\in\mathcal{M}$. The main result of this work asserts that there exists a subalgebra (named as ``automorphy subalgebra") on which each $\tau_{t}$ is an automorphism. Later, St\o rmer in \cite{stormer2007} studied multiplicative properties of a positive map on a von Neumann algebra 
and obtained more general results concerning the automorphic behaviour.  

Our aim in this paper is to study multiplicative properties of quantum channels acting on a matrix algebra $\Md$. It is well known that a quantum channel $\Ep:\Md\rightarrow\Md$ is represented by a  set of (non-unique) operators $\{k_{i}\}_{i=1}^{n}$ such that 
\begin{align*}
\Ep(x)=\sum_{j=1}^{n}k_{j}xk_{j}^*, \ \forall x\in \Md \  \ and \  \sum_{j=1}^{n}k_{j}^*k_{j}=1.
\end{align*}
The operators $\{k_{i}\}_{i=1}^{n}$ are known as \emph{Kraus operators}.
Our first result is: for a unital quantum channel, we have the equality of two sets   \[\Me=\mathcal{F}_{\Ep^*\circ\Ep} \ ,\]
where $\Me$ is the multiplicative domain of $\Ep$ and 
$\mathcal{F}_{\Ep^*\circ\Ep}$ denotes the fixed point set of $\Ep^*\circ\Ep$.
 Here $\Ep^*$ is the adjoint of $\Ep$ when $\Md$ is thought of as a Hilbert space with the Hilbert-Schmidt inner product $\langle a,b\rangle=\tr(ab^*)$ for all $a,b\in \Md$. The adjoint $\Ep^*$ satisfies the relation $\tr(\Ep(a)b)=\tr(a\Ep^*(b))$ for all $a,b\in \Md $. This result seems to be known before (see Theorem 10, 11 in \cite{c-j-k} and also \cite{kribs-spekkens}) but we present a different proof here. Exploiting the same relationship for powers of the channel, that is for $\Ep^n$, $n\geq 1$, we arrive at the chain of subalgebras with the following inclusion 
\begin{center} $\mathcal{M}_{\Ep}\supseteq\mathcal{M}_{\Ep^2}\supseteq
\cdots \supseteq \mathcal{M}_{\Ep^n}\supseteq\cdots$.
\end{center}
Since the underlying space is of finite dimension, this finding motivates us to predict the existence of a stabilising subalgebra which we denote by
 \[\Mi:=\bigcap_{n=1}^{\infty}\mathcal{M}_{\Ep^n}.\]
This subalgebra captures precisely the automorphic behaviour of $\Ep$ which is the content of the main theorem (\ref{main theroem}) of this paper. It turns out that $\Ep$ acts as a bijective homomorphism on $\Mi$ and also this set is the algebra generated by the eigen operators of $\Ep$ corresponding to the peripheral eigenvalues. The sub algebra $\Mi$ carries the intrinsic automorphic  attribute of a unital channel which also manifests the asymptotic behaviour of $\Ep$. Theorem \ref{main theroem} then sets the foundation of introducing the notion of \emph{multiplicative index} of a unital channel. It is the smallest $n\in \mathbb{N}$ such that $\mathcal{M}_{\Ep^n}=\Mi$. It turns out that the multiplicative index has important connection in quantum error correction which is described in Section \ref{sec:application}. 

   Much of the work presented in this paper is based on viewing a channel in the Schr\"odinger picture where trace preservation is assumed and exploited heavily. In a non-unital case, such a channel is realised as a unital completely positive map in the Heisenberg picture (that is, in the dual picture). In Section \ref{heisenberg pic}, we prove that the set of unital completely positive maps that have trivial multiplicative domains is cb dense in the class of all unital completely positive maps. Also a new result is obtained (Theorem \ref{pripheral eigenvectors}) which can be viewed as an extension to the Arveson's boundary theorem (\ref{Boundary Theorem}) on matrix algebras. 
%%%%%%%%%%%%%%%%%%%%%%

The paper is organised as follows. For a unital channel $\Ep:\Md\rightarrow\Md$ 
we start with Section \ref{Fixed Point and Multiplicative Domain section} that develops the techniques needed to prove 
the relationship between the fixed point set of $\Ep^*\circ\Ep$ and the multiplicative domain of $\Ep$. 
Some related corollaries are noted down. In Section \ref{mult index} we introduce the notion of \emph{multiplicative index}. 
The main theorem (Theorem \ref{main theroem}) characterises the stabilising subalegbra $\Mi$ in terms of the peripheral eigenvectors of $\Ep$. Also the multiplicative index is calculated for some quantum channels. Section \ref{irreducible section} is concerned with the multiplicative properties of irreducible and primitive quantum channels, types of  channels that appear to be very important in information theory. In Section \ref{heisenberg pic} the multiplicative nature of unital completely positive (not necessarily trace preserving) maps is  explored. The tools and techniques developed throughout the paper are exploited in Section \ref{sec:application} to demonstrate some applications in information theory, specifically in quantum error correction. Section \ref{sec:summary} contains the summary and some discussions on the topological aspects of sets with fixed multiplicative index. 

%%%%%%%%%%%%%%%%%%%%%%%%%%%%%%%%%%%%%%%%55 
 
\section{Fixed Point and Multiplicative Domain}\label{Fixed Point and Multiplicative Domain section}

We begin with some terminology and general theory of positive and completely positive maps on C$^*$-algebras which be needed for further discussion. The books \cite{paulsen} and \cite{stormer} are amongst many good references in this topic.

Let $\Ep:\mathcal{A} \rightarrow \mathcal{B}$ be a unital completely positive map of $C^*$ algebras $\mathcal{A}$ and $\mathcal{B}$  . The following sets are called the set of fixed points and the multiplicative domain respectively:
\begin{center}
$\mathcal{F}_{\Ep}=\{a\in \mathcal{A}:\Ep(a)=a\},$\\ 
\vspace{0.3cm}
$\mathcal{M}_{\Ep}=\{a\in \mathcal{A}:\Ep(ab)=\Ep(a)\Ep(b), \ \Ep(ba)=\Ep(b)\Ep(a) \ \forall b \in \mathcal{A}\}$.
\end{center}
Recall that a completely positive unital map $\Ep$ satisfies the Schwarz inequality $\Ep(aa^*)\geq \Ep(a)\Ep(a^*)$, for every $a\in \mathcal{A}$.
Choi (\cite{choi1}) showed,  for a unital completely positive map, the set  $\mathcal{M}_{\Ep}$ is same as the following set
\begin{center}
$\mathcal{S}=\{x\in \Md:\Ep(xx^*)=\Ep(x)\Ep(x^*), \Ep(x^*x)=\Ep(x^*)\Ep(x)\}$.
\end{center}
Recall that for a unital completely positive map $\Phi:\mathcal{A}\rightarrow \Bh$, the Stinespring dilation theorem says that there exist a Hilbert space $\mathcal{K}$, a bounded linear operator $V:\mathcal{H}\rightarrow \mathcal{K}$ and a $*$ homomorphism $\pi:\mathcal{A}\rightarrow \mathcal{B}(\mathcal{K})$ such that $\Phi(x)=V^*\pi(x)V$, for all $x \in \mathcal{A}$.
Furthermore, $||V||^2\leq ||\Phi(1)||=1$.
We first state a theorem regarding multiplicative domain of a unital completely positive (ucp for short!) map defined on a C$^*$-algebra which can be found as an exercise in \cite{paulsen}, see also \cite{bunce-salinas}
\begin{theorem}{\label{paulsen's exc}}
Let $\mathcal{A}$ be a unital $C^*$ algebra, and let $\Ep:\mathcal{A}\rightarrow \Bh$ be a ucp map with the  minimal Stinespring representation 
$(\pi,V,\mathcal{K})$. An element $a\in \mathcal{A}$ satisfies $\Ep(aa^*)=\Ep(a)\Ep(a)^*$ and $\Ep(a^*a)=
\Ep(a)^*\Ep(a)$ if and only if $V\mathcal{H}$ is a reducing subspace for $\pi(a)$. Moreover, the  collection of such elements is a $C^*$ sub-algebra of $\mathcal{A}$ and 
equals to the multiplicative domain of $\Ep$.
\end{theorem}
The following theorem provides useful  characterisations of projections and unitaries to belong to the multiplicative domain of a ucp map.
\begin{theorem}{\label{projection and unitary}}
Let $\Ep:\mathcal{A}\rightarrow \mathcal{B}$ be a ucp  map between unital $C^*$ algebras. Then
\begin{enumerate}
\item for a projection $p\in \mathcal{A}$, $p\in  \Me$ if and only if $\Ep(p)$ is
a projection.
\item for a unitary element $u\in \mathcal{A}$, $u\in \Me$ if and only if $\Ep(u)$ is a unitary element.
\end{enumerate}
\end{theorem}
\begin{proof}
\begin{enumerate}
\item If a  projection $p\in \Me$, then by Theorem \ref{paulsen's exc} $\Ep(p)=\Ep(p^2)=\Ep(p)\Ep(p)$.\\
Conversely, for a projection $p\in \mathcal{A}$, if $\Ep(p)^2=\Ep(p)$, then $\Ep(p^2)=\Ep(p)\Ep(p)$ and hence $p$ gives equality in the Schwarz inequality and by the Theorem \ref{paulsen's exc}, $p\in \Me$. 
\item If $u\in \Me$ and $u$ is a unitary, then $\Ep(1)=1=\Ep(uu^*)=\Ep(u)\Ep(u)^*$ and similarly the other direction.\\
Conversely, if $\Ep(u)$ is unitary for a unitary $u\in \mathcal{A}$, then it is easy to see that $u$ satisfies the equality in the Schwarz inequality as well and hence we get $u\in \Me$.
\end{enumerate}
\end{proof}
%%%%%%%%%%%%%%%%%%%%%%%%%%%%%%%%
 \noindent{\bf Assumptions.} 
\begin{enumerate} 
\item With the exception of Example \ref{free group example} all quantum channels considered henceforth in this paper are assumed to act on the matrix algebra $\mathcal M_d$.
\item Given a quantum channel $\Ep:\Md\rightarrow\Md$, we can identify its dual map or adjoint map $\Ep^*$ via the relation
$\tr(\Ep(a)b)=\tr(a\Ep^*(b)) \ \forall \ a, \ b \ \in \Md $ where 
 $\langle a,b\rangle=\tr(ab^*)$ for all $a,b \in \Md$,  defines an inner product on $\Md$ which makes $\Md$ a Hilbert space. This is known as Hilbert-Schmidt inner product. We will frequently denote the norm of an element $x\in \Md$, arising from this inner product as $\|x\|^2_{H.S}:=\langle x,x\rangle=\tr(xx^*)$.
 \end{enumerate}
We are now ready to state the following theorem. The result was known before (see \cite{c-j-k}, \cite{kribs-spekkens}) but we present a different proof here. The technique used in this proof will be used significantly throughout the rest of the paper.
\begin{theorem}{\label{mult.dom thm}}
Let $\Ep:\Md \rightarrow \Md$ be a unital quantum channel. Then
\begin{center}
 $\mathcal{M}_{\Ep}=\mathcal{F}_{(\Ep^{*}\circ \Ep)}$.
\end{center}
That is, the multiplicative domain of $\Ep$ equals to the fixed point set of $(\Ep^*\circ \Ep)$.
\end{theorem}
\begin{proof}
If $x \in \mathcal{M}_{\Ep}$, then $\Ep(xy)=\Ep(x)\Ep(y)$ and $\Ep(yx)=\Ep(y)\Ep(x)$ for all $y \in \Md$.\\
We then have, for any $z \in \Md$, 
\begin{center}
 $\langle x,z\rangle=\tr(z^*x)=\tr(\Ep(z^*x))=\tr(\Ep(z)^*\Ep
(x))=\langle \Ep(x),\Ep(z)\rangle .$ 
\end{center}
Invoking the adjoint relation, we have 
\begin{center}
$\langle x,z\rangle=\langle\Ep^{*}\circ\Ep(x),z\rangle$.
\end{center}
Since this happens for all $z \in \Md$, by the non-degeneracy of the pairing we have $\Ep^{*}\circ\Ep(x)=x$.\\
Conversely, if $x \in \Md$ is such that $\Ep^{*}\circ\Ep(x)=x$, then 
\begin{align*}
\langle x,x\rangle=\tr(x^*x)&=\tr(\Ep(x^*x))\\ 
&\geq \tr(\Ep(x)^{*}\Ep(x))\\
&=\langle\Ep(x),\Ep(x)\rangle\\
&=\langle \Ep^{*}\circ\Ep(x),x\rangle=\langle x,x\rangle.
\end{align*}
Where the inequality arises from from the Schwarz inequality for the ucp map $\Ep$. Since the extreme ends of the above equations are same, the inequalities become equality.
So we have $\tr(\Ep(x^*x))=\tr(\Ep(x)^{*}\Ep(x))$,
and hence by the faithfulness of the trace we have $\Ep(x^*x)=
\Ep(x)^{*}\Ep(x)$. Now by  Theorem \ref{paulsen's exc}, we conclude $x \in \mathcal{M}_{\Ep}$.
\end{proof}

Now we relate the multiplicative domain with the  commutant of product of Kraus operators of a unital channel.
\begin{corollary}[Commutant \ of \ the \ Kraus \ Operators]{\label{commutant}}
Let $\Ep:\Md\rightarrow\Md$ be a unital quantum channel with the Kraus representation:  $\Ep(x)=\sum_{j=1}^{n}a_{j}xa_{j}^*$, for all $x\in \Md$. Then an element $\rho \in \mathcal{M}_{\Ep}$ if and only if $\rho a_{i}^*a_{j}=a_{i}^*a_{j}\rho$ for all $i,j=1,2, \cdots, n$.
\end{corollary}
\begin{proof}
Let us note that if $\Ep$ is unital channel, then $\Ep^{*}\circ\Ep$ is a unital channel as well. Since $\Ep$ is completely positive and trace preserving,  $\Ep^*$ is positive and unital because $\tr(x)=\tr(x\cdot1)=\tr(\Ep(x)\cdot1)=\tr(x\Ep^*(1))$ for all $x \in \Bh$ and hence $\Ep^*(1)=1$. If $\Ep$ is completely positive, then so is $\Ep^*$(see \cite{stormer}). Also
since $\Ep$ is unital, $\Ep^*$ is trace preserving because $\tr(\Ep^*(x))=\tr(\Ep^*(x)1)=\tr(x\Ep(1))=\tr(x)$.
And hence $\Ep^*$ is a ucp and trace preserving map and hence the composition $\Ep^*\circ\Ep$ is the same.
Now if $\Ep(x)=\sum_{j=1}^{n}a_{j}xa_{j}^* \ \forall x \in \Md$, is a Kraus representation, a little calculation shows that $\Ep^*\circ\Ep$ is represented by
\begin{align*}
\Ep^*\circ\Ep(x)=\sum_{i,j=1}^{n}a_{i}^*a_{j}xa_{j}^*a_
{i}.
\end{align*}
And so the Kraus operators for $\Ep^*\circ\Ep$ are  $\{a_{i}^*a_{j}\}_{i,j}$. Now by the fixed point and commutant theorem (\cite{krbs}, Theorem 2.1) and by the theorem \ref{mult.dom thm} we get $\rho \in \mathcal{M}_{\Ep}$ if and only if 
\begin{center}
$\rho a_{i}^*a_{j}=a_{i}^*a_{j}\rho ,\forall
 i,j. $
\end{center}
%The theorem can be proved without the aid
%of the Theorem \ref{mult.dom thm} by a direct calculation. For the completeness, we sketch the proof.\\
%If $\rho \in \Me$, then By Theorem \ref{paulsen's exc},  $\Ep(\rho\rho^*)=\Ep(\rho)\Ep(
%\rho)^*$. Now using the unitality and trace preservation of $\Ep$ we get 
%\begin{align*}
%0&\leq Tr([a_{i}^*a_{j}\rho-\rho a_{i}^*a_{j}][a_{i}^*a_{j}\rho-\rho a_{i}^*a_{j}]^*)\\
%&\leq \sum_{i,j=1}^{n}Tr([a_{i}^*a_{j}\rho-\rho a_{i}^*a_{j}][a_{i}^*a_{j}\rho-\rho a_{i}^*a_{j}]^*)\\
%&=\sum_{i,j=1}^{n}Tr([a_{i}^*a_{j}\rho-\rho a_{i}^*a_{j}][\rho^*a_{j}^*a_{i}- a_{j}^*a_{i}\rho^*])\\
%&=\sum_{i=1}^{n}Tr[a_{i}^*\Ep(\rho\rho^*)a_{i}]+
%\sum_{i,}^{n}Tr[\rho(a_{i}^*Ia_{i})\rho^*
%]-\\
%&\sum_{i=1}^{n}Tr[a_{i}^*\Ep(\rho)a_{i}\rho^
%*]-\sum_{i=1}^{n}Tr[\rho a^*\Ep(\rho^*)a_{i}]\\
%&=Tr\Ep(\rho\rho^*)+Tr(\rho\rho^*)-Tr(\Ep(\rho)\Ep(\rho)^*)-Tr(\Ep(\rho)\Ep(\rho)^*)\\
%&=0
%\end{align*}
%And by the faithfulness of Trace function, we have
%$a_{i}^*a_{j}\rho-\rho a_{i}^*a_{j}=0$, $\forall i,j$, implying
%$a_{i}^*a_{j}\rho=\rho a_{i}^*a_{j}$, $\forall i,j$.\\
%Conversely, if Equation \ref{first equation} is satisfied by any $\rho$, then
%\begin{align*}
%\Ep(\rho)\Ep(\rho)^*&=(\sum_{i=1}^{n}a_{i}\rho a_{i}^*)(\sum_{j=1}^{n}a_{j}\rho^*a_{j}^*)\\
%&=\sum_{i,j=1}^{n}a_{i}\rho a_{i}^*a_{j}\rho^*a_{j}^*\\
%&=\sum_{j,i=1}^{n}a_{i} a_{i}^*a_{j}\rho\rho^*a_{j}^*\\
%&=\sum_{j=1}^{n}a_{j}\rho\rho^*a_{j}^*\\
%&=\Ep(\rho\rho^*)
%\end{align*}
%and hence by Theorem \ref{paulsen's exc}, $a \in \Me$.
\end{proof}
\begin{remark}
Trace preservation of $\Ep$ (equivalently, the unitality of $\Ep^*$) is the key factor of the above theorem.
Below we give an example where an element lies in the multiplicative domain of a ucp (but not trace preserving) map with a given Kraus representation does not satisfy the commutant condition. We use the example of a ucp map  that arose in \cite{bratteli}.\\
Let $\Phi:\mathcal{M}_{3}\rightarrow\mathcal{M}_{3}$ be a ucp map given by
\begin{center}
$\Phi\Bigg(\left[ \begin{array}{ccc} x_{11}&x_{12}&x_{13} \\ x_{21}&x_{22}&x_{23}\\ x_{31}&x_{32}&x_{33} \end{array}\right ]\Bigg)=\left[ \begin{array}{ccc} x_{11}&0&0 \\ 0&x_{22}&0\\ 0&0&\frac{x_{11}+x_{22}}{2} \end{array}\right ]$
\end{center}
Let us choose a set of Kraus operators for $\Phi$ as follows
\begin{center}
$k_{1}=\left[ \begin{array}{ccc} 1&0&0 \\ 0&0&0\\ 0&0&0 \end{array}\right ], k_{2}=\left[ \begin{array}{ccc} 0&0&0 \\ 0&1&0\\ 0&0&0 \end{array}\right ], k_{3}=\left[ \begin{array}{ccc} 0&0&0 \\ 0&0&0\\ \frac{1}{\sqrt{2}}&0&0 \end{array}\right ]$\\
and $k_{4}=\left[ \begin{array}{ccc} 0&0&0 \\ 0&0&0\\ 0&\frac{1}{\sqrt{2}}&0 \end{array}\right ]$
\end{center}
Now note that $a=\left[ \begin{array}{ccc} 1&0&0 \\ 0&1&0\\ 0&0&0 \end{array}\right ] \in \mathcal{M}_{\Phi}$ because $\Phi(a)=1$ where $1$ denote the $3\times3$ identity matrix and so $\Phi(a^*)=1$. Moreover $\Phi(aa^*)=\Phi(a)=1$ and $\Phi(a)\Phi(a^*)=1$.
And hence $a \in \mathcal{M}_{\Phi}$.\\
But $a$ does not commute with $k_{3}k_{1}^*=k_{3}$.
\end{remark}
In infinite dimension, Theorem \ref{commutant} is not true as the following example suggests. This example was first given in \cite{gudder} in a context of proving that the fixed point set is not necessarily equal to the commutant of Kraus operators of a channel. Since our context is similar, we use this example and expand it accordingly. Recall in infinite dimension, the issue of trace preserving needs to be addressed as the algebra in context might not have a finite trace. To this end, we follow the notion of \emph{quantum operation} given in \cite{gudder}. For a Hilbert space $\mathcal{H}$, if $\Bh$ is the set of bounded linear operators, then a linear map $\Ep:\Bh\rightarrow\Bh$ which is induced by a set of operators $\{a_{i}\}_{i=1}^{\infty}$ and defined as $\Ep(x)=\displaystyle\sum_{i=1}^{\infty} a_{i}xa_{i}^*$ is called \emph{trace preserving} if  $\tr(\Ep(b))=\tr b$, for every trace class operator $b$. It turns out that if $\displaystyle\sum_{i=1}^{\infty}a_{i}^*a_{i}=1$,  where the convergence is in the ultra-weak topology, then $\Ep$ is trace preserving. With these notations in hand, a quantum operation is a completely positive trace preserving map. 
\begin{example}\label{free group example}
Let $\mathbb{F}_{2}$ be the free group of two generators
$g_{1}, g_{2}$, with the identity element $e$. So $\mathbb{F}_{2}$ is a countable group. Let $\mathcal{H}$ be the complex separable Hilbert space 
\begin{center}
$\mathcal{H}=\{f:\mathbb{F}_{2}\rightarrow \mathbb{C}:
\sum_{x}|f(x)|^2 <\infty\}$.
\end{center}
Now define the following function for $x\in \mathbb{F}_{2}$,
\[ \delta_{x}(y)=   \left\{
\begin{array}{ll}
      1 & \text{if} \ y=x \\
      0 & \text{otherwise}.
\end{array} 
\right. \]
Then $\{\delta_{x}:x\in \mathbb{F}_{2}\}$ is an orthonormal basis for $\mathcal{H}$. It is well known that the group C$^*$-algebra $C^*(\mathbb{F}_{2})$ corresponding to the left regular representation 
$\Gamma:\mathbb{F}_{2}\rightarrow\Bh$ has a faithful trace
$\tau$, defined by $\tau(x)=\langle x\delta_{e},\delta_{e}\rangle$.\\
 Define two unitary operators on $\mathcal{H}$ by $u, v$ in the following way:
\begin{center}
$u(\delta_{x})=\delta_{g_{1}x}$ and $v(\delta_{x})=\delta_{g_{2}x}$ for all $x\in \mathbb{F}_{2}$.
\end{center}
Now define $\Ep:\Bh\rightarrow \Bh$ by 
\begin{center}
$\Ep(a)=\frac{1}{2}uau^*+\frac{1}{2}vav^*$
for all $a\in \Bh$.
\end{center}
Then $\Ep$  defines a ucp and trace preserving map.
 Let us call $\mathcal{M}$ the von Neumann algebra generated by $\{u^*v, v^*u, 1\}$.
 
Suppose an operator $b$ is defined as $b(\delta_{x})=\lambda_{x}\delta_{x}$ for all $x\in \mathbb{F}_{2}$ and $\lambda_{x}\in [0,1]$. 
The operator $b$ is positive and in the multiplicative domain $\Me$ of $\Ep$ if and only if  we have 
$\Ep(bb^*)=\Ep(b)\Ep(b)^*$ which yields for all $x\in \mathbb{F}_{2}$, $\Ep(bb^*)(\delta_{x})=\Ep(b)\Ep(b)^*(\delta_{x})$. Unwinding the definition of $\Ep$, we get 
\begin{center}
$\frac{1}{2} (ubb^{*}u^*+vbb^{*}v^*)(\delta_{x})=
\frac{1}{2}(ubu^*+vbv^*)\frac{1}{2}(ub^{*}u^*+vb^{*}v^*)(\delta_{x})$.
\end{center}  
Now applying the definition of $u, v$ and $b$, we get 
\begin{center}
\begin{equation}{\label{equ for mult.}}
(\frac{1}{2}\lambda^{2}_{g_{1}^{-1}x}+\frac{1}{2}\lambda^{2}_{g_{2}^{-1}x})=(\frac{1}{2}\lambda_{g_{1}^{-1}x}+\frac{1}{2}\lambda_{g_{2}^{-1}x})^{2} \ for \ all \ x\in \mathbb{F}_{2}.
\end{equation}
\end{center}
Now, if we assume $b\in \mathcal{M}'$, 
 and if $x=g^{-1}_{1}g_{2}y$ for some $y\in \mathbb{F}_{2}$, then we have, 
\begin{center}
$\lambda_{x}\delta_{x}=b\delta_{x}=b\delta_{g^{-1}_{1}
g_{2}y}=bu^*v\delta_{y}=u^*vb\delta_{y}=
\lambda_{y}u^*v\delta_{y}=\lambda_{y}\delta_{x}$.
\end{center}
And hence 
\begin{center}
\begin{equation}{\label{comm equ}}
\lambda_{g^{-1}_{1}g_{2}y}=\lambda_{y}, \ \forall y\in \mathbb{F}_{2}.
\end{equation}
\end{center}.

If $b$ is defined as $b(\delta_{x})=\lambda_{x}\delta_{x}$, where 
\[ \lambda_{x}=   \left\{
\begin{array}{ll}
      0 & \text{if $x$ ends with} \ g_{2}^{-1}\\
      1 & \text{if $x$ ends with} \ g^{-1}_{1}\\
      \frac{1}{2} & \text{otherwise}.
\end{array} 
\right. \]
Then one can check that $b$ satisfies equation \ref{equ for mult.} and hence $b \in \Me$. But $b\notin \mathcal{M}'$ because otherwise, we saw from Equation \ref{comm equ},  $\lambda_{g^{-1}_{1}g_{2}y}=\lambda_{y}$ for all $y\in \mathbb{F}_{2}$. Now putting $y=g^{-1}_{2}$ we have 
$1=\lambda_{g^{-1}_{1}}=\lambda_{g^{-1}_{2}}=0$, a contradiction.
\end{example}

\section{Multiplicative Index of a Unital Quantum Channel}{\label{mult index}}
In this section we discover more intrinsic properties  of a unital channel concerning its multiplicative behaviour. The relationship between spectral properties and the multiplicative nature will be 
explored but first we start with the following lemma which will be useful in subsequent discussion. 
\begin{lemma}{\label{primitive}}
For a unital quantum channel $\Ep$, we have 
\begin{center}
$\mathcal{M}_{\Ep^*\circ\Ep}=\mathcal{F}_
{\Ep^*\circ\Ep}=\mathcal{M}_{\Ep}$.
\end{center}
\end{lemma} 
\begin{proof}
The last equality is from Theorem \ref{mult.dom thm}.
We will establish the first equality of sets. Since the fixed point set is a subalgebra of the multiplicative domain, we automatically have $\mathcal{F}_{\Ep^*\circ\Ep}\subseteq \mathcal{M}_{\Ep^*\circ\Ep}$. For the converse part, let 
$a\in \mathcal{M}_{\Ep^*\circ\Ep}$. So it gives equality in Schwarz inequality for the map $\Ep^*\circ\Ep$ and we get 
\begin{center}
\begin{equation}\label{equ schwarz}
\Ep^*\circ\Ep(aa^*)=\Ep^*\circ\Ep(a)\Ep^*\circ\Ep(a^*)
\end{equation}
\end{center}
Now we have 
\begin{align*}
\tr(aa^*)&=\tr(\Ep^*\circ\Ep(aa^*))\\
&\geq \tr(\Ep^*[\Ep(a)\Ep(a^*)])\\
&\geq \tr(\Ep^*\Ep(a)\Ep^*\Ep(a^*))\\
&= \tr (\Ep^*\circ\Ep(aa^*))\\
&= \tr(aa^*)
\end{align*}
where the first two inequalities follow from Schwarz inequality for $\Ep$ and $\Ep^*$ and then we have used the relation in Equation \ref{equ schwarz}. Since the extreme ends of the above equation are same, we have equalities
in all the inequalities. This gives $\tr((\Ep^*\circ\Ep(aa^*))=\tr(\Ep^*[\Ep(a)\Ep(a^*)])$. Since $\Ep^*$ is trace preserving, we get $\tr(\Ep(aa^*))=\tr(\Ep(a)\Ep(a^*))$. Hence by the faithfulness of trace, we get 
$\Ep(aa^*)=\Ep(a)\Ep(a^*)$ and $a\in \mathcal{M}_{\Ep}=\mathcal{F}_{\Ep^*\circ\Ep}$.   
\end{proof}
Given a linear map $\Phi:\Md\rightarrow\Md$, the spectrum of $\Phi$ which is denoted by $\rm{Spec(\Phi)}$, is defined as
\[\rm{Spec(\Phi)}=\{\lambda\in \mathbb{C}: (\lambda 1-\Phi) \ \text{is not invertible on} \ \Md\},\]
where $1$ denotes the identity operator on $\Md$. Recall that the spectral radius of $\Phi$ which is denoted as $r(\Phi)$, is defined as
\[r(\Phi)=\sup\{|\lambda|:\lambda\in \rm{Spec(\Phi)}\}.\]
It follows that if $\Phi$ is a unital positive map, then $r(\Phi)\leq 1$ and hence all eigenvalues lie in the unit disc of the complex plane (see Proposition 6.1 in \cite{wolf}). For a unital channel $\Phi$ the
set $\rm{Spec(\Phi)}\cap \mathbb{T}$ is called the set of \emph{peripheral eigenvalues}, where $\mathbb{T}$ is the unit circle in the complex plane.
If an element $a\in \Md$ satisfies the relation $\Phi(a)=\lambda a$ for $|\lambda|=1$, then $a$ is called a \emph{peripheral eigenvector} of $\Phi$. Note that the fixed point set of $\Phi$, that is $\mathcal{F}_{\Phi}$, is the set of all peripheral eigenvectors corresponding to the eigenvalue $1$.
With these terminology in hand, we are ready to note down the corollary to Lemma \ref{primitive}.  
 
\begin{corollary}
For any unital channel $\Ep$, the channel $\Ep^*\circ\Ep$ does not have any peripheral eigenvalue other than 1.
\end{corollary}
\begin{proof}
We will show that any peripheral eigenvector corresponding to a peripheral eigenvalue of a unital channel $\Phi$ is in the multiplicative domain $\mathcal{M}_{\Phi}$. To this end,  
let $\lambda(\neq 1)$ is a peripheral eigenvalue of $\Phi$, that is $|\lambda|=1$ and suppose $a\in \Md$ is such that $\Phi(a)=\lambda a$. Now it follows form the positivity of $\Phi$ that $\Phi$ preserves the $*$-operation, that is, $\Phi(x^*)=\Phi(x)^*$, for every $x\in \Md$. Hence $\Phi(a^*)=\bar{\lambda}a^*$.
We get
\begin{align*}
\Phi(aa^*)\geq \Phi(a)\Phi(a^*)&=\lambda a\bar{\lambda}a^*\\
&=aa^*
\end{align*}
Using the trace preservation and faithfulness of trace, we get $\Phi(aa^*)=aa^*$ and hence $a\in \mathcal{M}_{\Phi}$. 

Now for a unital channel $\Ep$, if $a$ is a peripheral eigenvector of $\Ep^*\circ\Ep$ corresponding to a peripheral eigenvalue $\lambda(\neq 1)$, then $a\in \mathcal{M}_{\Ep^*\circ\Ep}$. But Lemma \ref{primitive} asserts that $\mathcal{M}_{\Ep^*\circ\Ep}=\mathcal{F}_
{\Ep^*\circ\Ep}$, which implies $a$ is an eigenvector
of $\Ep^*\circ\Ep$ corresponding to the eigenvalue 1. Hence we get a contradiction.  
\end{proof}
The next lemma sets the foundation of the concept of the multiplicative index of a quantum channel. It gives the description of the multiplicative domain of a composition of two quantum channels.
\begin{lemma}{\label{composition}}
Let $\Ep=\Phi\circ\Psi$ where $\Phi,\Psi$ are two unital quantum channels. Then 
\begin{center}
$\mathcal{M}_{\Ep}=\{a\in \mathcal{M}_{\Psi}:\Psi(a)\in \mathcal{M}_{\Phi}\}$.
\end{center}
\end{lemma}
\begin{proof}
If $a\in \mathcal{M}_{\Psi}$ such that $\Psi(a)\in \mathcal{M}_{\Phi}$, then
\begin{center} $\Ep(aa^*)=\Phi(\Psi(aa^*))=
\Phi(\Psi(a)\Psi(a^*))=\Phi\circ\Psi(a)\Phi\circ\Psi
(a^*)=\Ep(a)\Ep(a^*)$.
\end{center}
and hence $a\in \mathcal{M}_{\Ep}$.

Conversely, let $a\in \mathcal{M}_{\Ep}$. Then 
\begin{align*}
\Ep(a)\Ep(a^*)=\Ep(aa^*)&=\Phi\circ\Psi(aa^*)\\
&\geq \Phi(\Psi(a)\Psi(a^*))\\
&\geq \Phi(\Psi(a))\Phi(\Psi(a^*))\\
&=\Ep(a)\Ep(a^*).
\end{align*}
Hence all the inequalities must be equalities and we first get $\Phi\circ\Psi(aa^*)=\Phi(\Psi(a)\Psi(a^*))$. Now the trace preservation property of $\Phi$ would imply $\Psi(aa^*)=\Psi(a)\Psi(a^*)$ and we get $a\in \mathcal{M}_{\Psi}$.
Using the second inequality it's immediate that $\Psi(a)
\in \mathcal{M}_{\Phi}$. 
\end{proof}
Now we can proceed in exploring the multiplicative domain of compositions of a channel $\Ep$ with itself.
\begin{corollary}{\label{power relation}}
For $\Ep^n=\Ep\circ\cdots\circ\Ep$ ($n$-times, $n\in \mathbb{N}$), we have 
\begin{center}
$\mathcal{M}_{\Ep^n}=\{a\in \mathcal{M}_{\Ep^{n-1}}:\Ep(a)\in \mathcal{M}_{\Ep^{n-1}}\}$.
\end{center}
In particular, we have the following inclusion
\begin{center} $\mathcal{M}_{\Ep}\supseteq\mathcal{M}_{\Ep^2}\supseteq
\cdots \supseteq \mathcal{M}_{\Ep^n}\supseteq\cdots$.
\end{center}
\end{corollary}
Now since multiplicative domain of $\Ep^n$ is a C$^*$-algebra for each $n\in \mathbb{N}$ and the underlying subspace is of finite dimension, the above decreasing chain of subalgebras will stabilise at a fixed subalgebra. Let us denote this sub-algebra as $\Mi$ i.e 
\begin{align*}
\Mi=\bigcap_{n=1}^{\infty}\mathcal{M}_{\Ep^n}
\end{align*}
We will see that on this subalgebra $\Mi$, $\Ep$ acts as an automorphism. Also, it is not necessarily true that if $a\in \Me$, then $\Ep(a)\in \Me$. However, it will be evident that if $a\in \Mi$, then $\Ep(a)\in \Mi$. Note that St{\o}rmer in \cite{stormer2007} deals with a set related to a positive map, which he calls  the \emph{multiplicative core}
and proves that the positive linear map when restricted to this set, is a Jordan automorphism. In our context, the underlying space is of finite dimension and the linear map is completely positive which is stronger than positivity. The indispensable effect of the adjoint of a quantum channel  
in this whole discussion about multiplicative property, might be traced back to the work of K{\"u}mmerer in \cite{kummerer}, Sec. 3.2 where the author deals with dilations of asymptotic automorphic dynamical systems.
Now we state and prove the main theorem of this section.
\begin{theorem}{\label{main theroem}}
Let $\Ep:\Md\rightarrow\Md$ be a unital quantum channel.
Then
\begin{enumerate}
\item There exists a subalgebra of $\Md$, namely $\Mi$,  upon which $\Ep$ acts as a bijective homomorphism with the inverse being the adjoint $\Ep^*$.
\item $\Md$ decomposes into two orthogonal subspaces as $\Md=\Mi\bigoplus\Mi^\perp$, where the orthogonality is with respect to the Hilbert-Schmidt inner product. Moreover, we have a precise description for the set $\Mi^\perp$ given by
\begin{center}
$\Mi^\perp=\{x\in \Md:\displaystyle\lim_{n\to\infty}\|\Ep^n(x)\|=0\}$.
\end{center}
\item The spectrum of this automorphism is equal to the peripheral spectrum of $\Ep$, that is  $\rm{Spec(\Ep|_{\Mi})=Spec(\Ep)\cap \mathbb{T}}$. Moreover, if $\mathcal{N}_{\Ep}=\{a\in \Md:\Ep(a)=\lambda a, |\lambda|=1\}$, then $\Mi$ is the algebra generated by $\Ne$. 
\end{enumerate}
\end{theorem}
\begin{proof}
$1$. We will first show that $\Ep(\Mi)\subseteq\Mi$. To see this, let $a\in \Mi$. For any $k\in \mathbb{N}$, we have 
\begin{align*}
\Ep^{k+1}(aa^*)=\Ep^{k}(\Ep(aa^*))&=\Ep^{k}(\Ep(a)\Ep(a^*))\\
&\geq \Ep^{k}\Ep(a)\Ep^{k}\Ep(a^*)\\
&=\Ep^{k+1}(a)\Ep^{k+1}(a^*)\\
&=\Ep^{k+1}(aa^*).
\end{align*}
Where we have just used the Schwarz inequality for the map $\Ep^{k}$ and the equality follows because of the fact that $a\in \Mi=\bigcap_{n=1}^{\infty}\mathcal{M}_{\Ep^n}$.
Clearly, we have equality in all the inequalities and we obtain $\Ep^{k}(\Ep(a)\Ep(a^*))=\Ep^{k}(\Ep(a))\Ep^{k}(\Ep(a^*))$. Hence $\Ep(a)\in \mathcal{M}_{\Ep^k}$ for every $k\in \mathbb{N}$ and we get $\Ep(a)\in \Mi$.

Injectivity follows easily; let $a\in \Mi$ such that $\Ep(a)=0$, then
\begin{center} $\tr(aa^*)=\tr(\Ep(aa^*))=\tr(\Ep(a)\Ep(a^*))=0$,
\end{center}
which forces $a=0$. So $\Ep:\Mi\rightarrow\Mi$ is an injective linear map. Since $\Mi$ is finite dimensional vector space, the rank-nullity theorem asserts that $\Ep$ is surjective. \\
Now we prove that the inverse of $\Ep|_{\Mi}$ is $\Ep^*$. From Theorem \ref{mult.dom thm}, it is evident that on $\Mi\subseteq\Me$, $\Ep^*\circ\Ep=1$ where $1$
is the identity operator. As $\Ep$ is bijective on $\Mi$, for any $a\in \Mi$, there exists an element  $b\in \Mi$ such that $\Ep(b)=a$. Applying the adjoint both sides we get $\Ep^*\Ep(b)=\Ep^*(a)$. But $\Mi$ is a subset of the fixed point of $\Ep^*\circ\Ep$ and hence we get $b=\Ep^*(a)$. As $a$ was arbitrary, we have proved $\Ep^*(\Mi)\subseteq\Mi$. One can show similarly as was shown for $\Ep$, that $\Ep^*$ is a bijective homomorphism on $\Mi$. Now to show $\Ep^*$ is the right inverse, we let $a\in \Mi$ and find a $b\in \Mi$ such that $\Ep(b)=a$ and hence $\Ep\circ\Ep^*(a)=\Ep\circ\Ep^*(\Ep(b))=\Ep(b)=a$. So on $\Mi$, we get $\Ep\circ\Ep^*=\Ep^*\circ\Ep=1$.

2. The decomposition of $\Md$ is now clear since $\Mi$ is an invariant subspace for both $\Ep$ and $\Ep^*$ as was shown in part 1 of the proof. 
Now to get the characterisation of $\Mi^\perp$ we follow the method taken by St{\o}rmer in \cite{stormer2007}. Let $a\in \Mi^\perp$. Then we compute for any $n\in \mathbb{N}$, 
\begin{align*}
||\Ep^{n+1}(a)||^2_{H.S}&=\tr(\Ep^{n+1}(a)\Ep^{n+1}(a^*))\\
&\leq \tr(\Ep(\Ep^n(a)\Ep^n(a^*)))\\ 
&=\tr(\Ep^n(a)\Ep^n(a^*))\\
&=||\Ep^n(a)||^2_{H.S}
\end{align*}
So $\{||\Ep^{n}(a)||^2_{H.S}\}_{n=1}^{\infty}$ is a decreasing sequence and suppose $\Ep^n(a)\rightarrow a_{0}$ in the H.S sense. We will show that $a_{0}\in \Mi$. 
Since $\{||\Ep^{n}(a)||^2_{H.S}\}_{n=1}^{\infty}$ is decreasing, we have 
\begin{center}
\begin{equation}{\label{decresing}}
||\Ep^{n}(a)||^2_{H.S}-||\Ep^{n+1}(a)||^2_{H.S}\rightarrow 0 \ as \  n\rightarrow\infty.
\end{equation} 
\end{center}
By Schwarz inequality we see that 
\begin{center}
$\Ep(\Ep^n(a)\Ep^n(a^*))-\Ep^{n+1}(a)\Ep^{n+1}(a^*)\geq 0$
\end{center}
Taking trace and using Equation  \ref{decresing}, we obtain $\Ep(\Ep^n(a)\Ep^n(a^*))=\Ep^{n+1}(a)\Ep^{n+1}(a^*)$ as $n\rightarrow \infty$ i.e $\lim_{n\to\infty}\Ep^n(a)\in \Me$. Using this argument 
repeatedly we can show $\displaystyle\lim_{n\to\infty}\Ep^n(a)\in \mathcal{M}_{\Ep^k}$ for every $k\geq 1$. And we obtain 
$\displaystyle\lim_{n\to\infty}\Ep^n(a)\in \Mi$. Since the underlying space is a finite dimensional C$^*$-algebra, the weak limit is the norm limit and because $\Mi$ is a C$^*$-algebra, we find $a_{0}\in \Mi$. 

Lastly, if $a\in \Mi^\perp$, then $\Ep^k(a)\in \Mi^\perp$ for any $k\geq1$. Indeed, note that for every such $k$, $\Ep^k$ is bijective on $\Mi$ and for any $b\in \Mi$, there exists an element  $c\in \Mi$ such that $\Ep^k(c)=b$. So for any $b\in \Mi$, we have 
\begin{center}
$\tr(\Ep^k(a)b)=\tr(\Ep^k(a)\Ep^k(c))=\tr(\Ep^k(ac))=\tr(ac)
=0$.
\end{center} 
Hence $\Ep^k(a)\in \Mi^\perp$ for all $k\geq1$.
So we have $a_{0}\in \Mi\cap\Mi^\perp$, which forces $a_{0}=0$. Hence we get 
\begin{center}
$\Mi^\perp=\{x\in \Md:\displaystyle\lim_{n\to\infty}||\Ep^n(x)||_{H.S}=0\}$.
\end{center}
As $||\cdot||\leq ||\cdot||_{H.S}$, we have the desired description of the set.

3. We first prove that the eigen operators corresponding to the peripheral eigenvalues algebraically span the set $\Mi$.
It is not hard to see that $\Ne\subseteq \Mi$. Indeed if $\Ep(a)=\lambda a$ where $|\lambda|=1$, then we get $\Ep^k(a)=\lambda^k a$, for any $k\geq1$. Hence $\Ep^k(aa^*)\geq \Ep^k(a)\Ep^k(a^*)=\lambda^k\bar{\lambda}^kaa^*=aa^*$.
Taking trace and using the faithfulness of trace, we get 
$\Ep^k(aa^*)=\Ep^k(a)\Ep^k(a^*)$. Hence it follows that 
$a\in \mathcal{M}_{\Ep^k}$ for all $k\geq1$ and subsequently, $a\in \Mi$. Hence the algebra generated by $\Ne$ is contained in $\Mi$.

Conversely, the map $\Ep:\Mi\rightarrow\Mi$ is a bijection which satisfies $\Ep\circ\Ep^*=\Ep^*\circ\Ep=1$ that is a unitary operator on the Hilbert subspace $\Mi$ equipped with the Hilbert-Schmidt inner product. In particular, $\Ep$ is a normal operator and hence diagonalisable that is there exists a basis of eigenvectors
of $\Ep$ which span the entire space $\Mi$. Now if $a\in \Mi$ such that $\Ep(a)=\lambda a$, then $\Ep(aa^*)=\Ep(a)\Ep(a^*)=\lambda\bar{\lambda}aa^*$. Taking trace in both sides we see that $\lambda\bar{\lambda}=|\lambda|^2=1$ that is $a$ is one of the peripheral eigen operators of $\Ep$. So, $\Mi$ is spanned by the eigen operators corresponding the peripheral eigenvalues of $\Ep$ showing $\Mi$ is contained in the algebra spanned by $\Ne$ and hence we get the required equality.
It is now obvious that $\rm{Spec(\Ep|_{\Mi})=Spec(\Ep)\cap\mathbb{T}}$.
\end{proof}
\begin{remark}
Theorem \ref{main theroem} displays the stable  behaviour of the channel on the algebra generated by the peripheral eigen-operators and on its complementary part, $\Ep$ asymptotically approaches to zero. In \cite{arveson2}, similar results were obtained for a unital completely positive maps (not necessarily trace preserving) where the given map acts as an isometry on the operator system spanned by $\Ne$. See also \cite{veronica} for the related discussion on quantum Markov semigroups. 
 Our attention has been finite dimensional C$^*$-algebra and the maps $\Ep$ were unital completely positive and trace preserving. The existence of the adjoint $\Ep^*$ and subsequently the identity $\Me=\mathcal{F}_{\Ep^*\circ\Ep}$ helps to have a very different approach to this topic from the above mentioned cases.
\end{remark}
\begin{remark}
The decreasing chain of subalgebras 
\begin{center} $\mathcal{M}_{\Ep}\supseteq\mathcal{M}_{\Ep^2}\supseteq
\cdots \supseteq \mathcal{M}_{\Ep^n}\supseteq\cdots$
\end{center}
must stabilise at a finite stage because the underlying 
space is of finite dimension. Call  $\kappa$ to be the smallest number for which  we have $\mathcal{M}_{\Ep^n}=\mathcal{M}_{\Ep^\kappa}$ , $n\geq \kappa$. It is evident that $\Mi=\mathcal{M}_{\Ep^\kappa}$. This $\kappa$ is uniquely determined for every quantum channel $\Ep$ and can be used to differentiate between two quantum channels. We call this $\kappa$ as the \emph{multiplicative index} of $\Ep$. 
\end{remark}
Now we note down the following corollary which will be useful in future reference. Note that a channel is called \emph{normal} or \emph{diagonalisable} if $\Ep^*\circ\Ep=\Ep\circ\Ep^*$.    
\begin{corollary}{\label{diagonalisable}}
If a unital channel $\Ep$ satisfies $\Ep^*\circ\Ep=\Ep\circ\Ep^*$, then $\Mi=\Me$. Hence the multiplicative index of such channels is 1.
\end{corollary}
\begin{proof}
If $a\in \Me=\mathcal{F}_{\Ep^*\circ\Ep}$, then $\Ep^*\circ\Ep(a)=a$ and applying $\Ep^*\circ\Ep$ once again we get $\Ep^*\circ\Ep(\Ep^*\circ\Ep(a))=a$. Now using the commutativity, we get $(\Ep^*)^2\circ\Ep^2(a)=a$. We will show $\mathcal{M}_{\Ep^2}=\Me$. 

For $a\in \Me$, we have
\begin{align*}
\tr(aa^*)=\tr(\Ep^2(aa^*))&=\tr(\Ep(\Ep(a)\Ep(a^*)))\\
&\geq \tr(\Ep^2(a)\Ep^2(a^*))\\
&=\tr((\Ep^*)^2\Ep^2(a)a^*)\\
&=\tr(aa^*).
\end{align*}
So the inequality must be an equality and hence we get
$\Ep(a)\in \Me$ which yields $a\in \mathcal{M}_{\Ep^2}$. Now by Corollary \ref{power relation}, we get $\mathcal{M}_{\Ep^2}=\Me$. This process can be repeated for every $n\in \mathbb{N}$ to get $a\in \mathcal{M}_{\Ep^n}$, which forces $a\in \Mi$.
\end{proof}
Recall that the qubit Pauli operators are described by the following $2\times 2$ matrices
\[I=\left[\begin{array}{cc}
1&0\\
0&1
\end{array}\right], X=\left[\begin{array}{cc}
0&1\\
1&0
\end{array}\right], Y=\left[\begin{array}{cc}
0&-i\\
i&0
\end{array}\right],Z=\left[\begin{array}{cc}
1&0\\
0&-1
\end{array}\right] .\]
Channels whose Kraus decompositions consist of Pauli operators are called Pauli Channels. The generalised Pauli channels in dimension $d$ consists of random mixtures of unitaries in the discrete Weyl-Heisenberg representation. The Pauli or generalised Pauli channels are diagonalisable (\cite{ergodic}) and hence they all have multiplicative index 1. 
\begin{remark}
 By Theorem \ref{main theroem}, we get a decomposition of the unital channel $\Ep$ as follows  $\Ep=\begin{pmatrix}
\Ep_{0} & 0\\
0  & \Ep_{1}\\
\end{pmatrix}$,
where $\Ep_{0}$ is the automorphism on $\Mi$ and $\Ep_{1}$ is another quantum channel. This decomposition of $\Ep$ gives more information than the \emph{Jordan decomposition} for $\Ep$. Recall that a linear operator $\Phi:\Md\rightarrow\Md$, regarded as an element of $\mathcal{M}_{d^2}$, admits a \emph{Jordan decomposition} of the form 
\begin{align*}
\Phi=P\big (\bigoplus J_{k}(\lambda_{k}) \big )P^{-1}, \ 
J_{k}(\lambda)=\begin{pmatrix}
\lambda & 1 &  \\
 & \ddots & 1 \\
 & & \lambda
\end{pmatrix},
\end{align*}
where the $J_{k}$'s are Jordan blocks of size $d_{k}$ and $\sum_{k}d_{k}=d^2$. From the \emph{Jordan decomposition} it follows that each of the peripheral eigenvalues for a trace preserving or unital positive map $\Phi$, has one-dimensional Jordan blocks (see for example \cite{wolf}, Proposition 6.2). Although the Jordan decomposition is very useful in studying eigen values and their locations, the multiplicative nature of a channel cannot be investigated by just looking at the Jordan decomposition.
\end{remark}
 
Now we enlist some examples showing the different values of the multiplicative index $\kappa$ of unital channels and the corresponding algebras $\Mi$. In what follows $e_{1}, \cdots, e_{d}$ denote the standard basis of $\mathbb{C}^d$. For any two vectors
$\xi,\eta\in \mathbb{C}^d$, the rank one operator $\xi\eta^*:\mathbb{C}^d\rightarrow\mathbb{C}^d$ is defined by $\xi\eta^*(x)=\langle x,\eta\rangle \xi$, for all $x\in \mathbb{C}^d$.
\begin{example}\label{unitary}
If $\Ep$ is a unitary channel that is $\Ep(x)=uxu^*$ for some unitary $u$ and for all $x\in \mathcal{M}_{d}$, then $\Ep$ is multiplicative in the whole domain and $\Mi=\Md$ and $\kappa=1$.
\end{example}
\begin{example}\label{discrete four.for 3}
Let $\omega\in \mathbb{C}$ be such that $\omega^{3}=1$. Define a quantum channel $\Ep:\mathcal{M}_{3}\rightarrow\mathcal{M}_{3}$ in the following way:
\[\Ep(x)=\displaystyle\sum_{j=1}^{3}s_{j}xs_{j}^*,\]
for all $x\in \mathcal{M}_{3}$. Where 
\[s_{1}=\frac{1}{\sqrt{3}}\left[\begin{array}{ccc}
1 & 0 & 0\\
1 & 0 & 0\\
1 & 0 & 0\\
\end{array}\right], \ s_{2}=\frac{1}{\sqrt{3}}\left[\begin{array}{ccc}
0 & 1 & 0\\
0 & \omega & 0\\
0 & \omega^2 & 0\\
\end{array}\right], \ s_{3}=\frac{1}{\sqrt{3}}\left[\begin{array}{ccc}
0 & 0 & 1\\
0 & 0 & \omega^2\\
0 & 0 & \omega\\
\end{array}\right].\]

Calculation shows that 
\[\Ep^{*}\circ\Ep(x)=\displaystyle\sum_{j=1}^{3}a_{j}x
a_{j}^*,\]
where  $a_{1}=\left[\begin{array}{ccc}

1 & 0 & 0\\
0 & 0 & 0\\
0 & 0 & 0\\
\end{array}\right], \ a_{2}=\left[\begin{array}{ccc}
0 & 0 & 0\\
0 & 1 & 0\\
0 & 0 & 0\\
\end{array}\right], \ a_{3}=\left[\begin{array}{ccc}
0 & 0 & 0\\
0 & 0 & 0\\
0 & 0 & 1\\
\end{array}\right].$

\vspace{10pt}
 Clearly $\Me=\mathcal{F}_{\Ep^*\circ\Ep}=\Bigg
 \{\left[\begin{array}{ccc}
a & 0 & 0\\
0 & b & 0\\
0 & 0 & c\\
\end{array}\right] : a,b,c \in \mathbb{C}\Bigg\}$, which is the algebra of diagonal matrices, a commutative C$^*$-algebra.

Now $\Ep^2(x)=\displaystyle \sum_{j=1}^{9}b_{j}xb_{j}^*$,
\newline
where 
\[b_{1}=\frac{1}{3}\left[\begin{array}{ccc}
1 & 0 & 0\\
1 & 0 & 0\\
1 & 0 & 0\\
\end{array}\right], b_{2}=\frac{1}{3}\left[\begin{array}{ccc}
1 & 0 & 0\\
\omega & 0 & 0\\
\omega^2 & 0 & 0\\
\end{array}\right], b_{3}=\frac{1}{3}\left[\begin{array}{ccc}
1 & 0 & 0\\
\omega^2 & 0 & 0\\
\omega & 0 & 0\\
\end{array}\right],\]
\[b_{4}=\frac{1}{3}\left[\begin{array}{ccc}
0 & 1 & 0\\
0 & 1 & 0\\
0 & 1 & 0\\
\end{array}\right], b_{5}=\frac{1}{3}\left[\begin{array}{ccc}
0 & \omega & 0\\
0 & \omega^2 & 0\\
0 & 1 & 0\\
\end{array}\right], b_{6}=\frac{1}{3}\left[\begin{array}{ccc}
0 & \omega^2 & 0\\
0 & \omega & 0\\
0 & 1 & 0\\
\end{array}\right],\]
\[b_{7}=\frac{1}{3}\left[\begin{array}{ccc}
0 & 0 & 1\\
0 & 0 & 1\\
0 & 0 & 1\\
\end{array}\right], 
b_{8}=\frac{1}{3}\left[\begin{array}{ccc}
0 & 0 & \omega^2\\
0 & 0 & 1\\
0 & 0 & \omega\\
\end{array}\right], b_{9}=\frac{1}{3}\left[\begin{array}{ccc}
0 & 0 & \omega\\
0 & 0 & 1\\
0 & 0 & \omega^2\\
\end{array}\right].\]
Actually $\Ep^2(x)=\frac{\tr(x)}{3}1$, for all $x\in \mathcal{M}_{3}$.
We have $\mathcal{M}_{\Ep^2}=\mathcal{A}'$, where $\mathcal{A}$ is the algebra generated by the set $\{b_{i}^*b_{j}:i,j=1,\cdots, 9\}$. Now Since $\mathcal{A}$ generates the full matrix algebra $\mathcal{M}_{3}$, we get $\mathcal{A}'=\mathbb{C}1$.
So $\mathcal{M}_{\Ep^2}=\mathbb{C}1$. Hence for every $n\geq 2$, $\mathcal{M}_{\Ep^n}=\mathbb{C}1$, resulting $\Mi=\mathcal{M}_{\Ep^2}=\mathbb{C}1$. Hence the multiplicative index $\kappa=2$ and also we have found $\Mi\subset \Me$.
\end{example}
\begin{example}
The above example is a particular case $(d=3)$ of a more general example that can be constructed on $\mathcal{M}_{d}$ (see example after Theorem 13 in \cite{ergodic}), where $\Ep:\Md\rightarrow\Md$ is defined by
\[\Ep(x)=\displaystyle \sum_{j=1}^{d}s_{j}xs_{j}^*,\]
where $s_{k}=f_{k}e_{k}^*$, where $e_{k}$'s are standard basis of $\mathbb{C}^d$ and $f_{k}$'s are the Fourier basis $f_{k}=\frac{1}{\sqrt{d}}\displaystyle
\sum_{j=1}^{d}e^{{2\pi ikj}/d}e_{k}$.
It follows that $\Ep^2(x)=\tr(x)\frac{1}{d}$ and hence has trivial multiplicative domain and hence $\Mi=\mathcal{M}_{\Ep^2}=\mathbb{C}1$. However the multiplicative domain of $\Ep$ is generated by the rank one projections $e_{1}{e_{1}}^*, \cdots, e_{d}{e_{d}}^*$. So $\Mi\subset\Me$ and $\kappa(\Ep)=2$ for every dimension $d$.  

\end{example}
\begin{example}\label{kappa three}
 Now we construct a channel on $\mathcal{M}_{3}$ with $\kappa=3$.
 Let $\Ep:\mathcal{M}_{3}\rightarrow\mathcal{M}_{3}$ be given by \[
\Ep(x)= \displaystyle \sum_{j=1}^{3}s_{j}xs_{j},
\]
where $s_{1}=\frac{1}{\sqrt{2}}\left[\begin{array}{ccc}
0 & 0 & 1\\
0 & 0 & 1\\
0 & 0 & 0\\
\end{array}\right]$, $s_{2}=\frac{1}{\sqrt{2}}\left[\begin{array}{ccc}
0 & 1 & 0\\
0 & -1 & 0\\
0 & 0 & 0\\
\end{array}\right]$ and $s_{3}=\left[\begin{array}{ccc}
0 & 0 & 0\\
0 & 0 & 0\\
1 & 0 & 0\\
\end{array}\right]$.

It can be seen that $\Ep^*\circ\Ep(x)=\displaystyle \sum_{j=1}^{3}a_{j}xa_{j}$, where $a_{j}=e_{j}{e_{j}}^*$, for $j=1,2,3$, that is the rank one projections on standard basis $\{e_{1}, e_{2}, e_{3}\}$. So 
\[\Me=\mathcal{F}_{\Ep^*\circ\Ep}=\Bigg
 \{\left[\begin{array}{ccc}
a & 0 & 0\\
0 & b & 0\\
0 & 0 & c\\
\end{array}\right] : a,b,c \in \mathbb{C}\Bigg\}.\]
Now, take $p=\left[\begin{array}{ccc}
0 & 0 & 0\\
0 & 0 & 0\\
0 & 0 & 1\\
\end{array}\right]$, then it can be checked that 
$\Ep(p)=\frac{1}{2}\left[\begin{array}{ccc}
1 & 1 & 0\\
1 & 1 & 0\\
0 & 0 & 0\\
\end{array}\right]$, and hence $\Ep(p)\not\in \Me$.
Since $p\in \Me$ and $\Ep(p)\not\in \Me$, we get $p\not\in\mathcal{M}_{\Ep^2}$. Indeed one checks that
\[\mathcal{M}_{\Ep^2}=\text{span}\Bigg\{\left[\begin{array}{ccc}
0 & 0 & 0\\
0 & 1 & 0\\
0 & 0 & 1\\
\end{array}\right],\left[\begin{array}{ccc}
1 & 0 & 0\\
0 & 0 & 0\\
0 & 0 & 0\\
\end{array}\right]\Bigg\}.\]

Now for $q=\left[\begin{array}{ccc}
1 & 0 & 0\\
0 & 0 & 0\\
0 & 0 & 0\\
\end{array}\right]$, it turns out that $\Ep^2(q)\not\in\mathcal{M}_{\Ep^2}$ and hence $q\not\in \mathcal{M}_{\Ep^3}$. Indeed we get 
$\mathcal{M}_{\Ep^3}=\mathbb{C}1\subset \mathcal{M}_{\Ep^2}$, which yields 
$\Mi=\mathcal{M}_{\Ep^3}=\mathbb{C}1$ and hence $\kappa(\Ep)=3$.
\end{example}
%%%%%%%%%%%%%%%%%%%%
It is worth mentioning at this point that the value of the multiplicative index $\kappa$ of a unital channel $\Ep:\Md\rightarrow\Md$ can not range from $1$ to $d^2$. The reason is the reduction of dimension of maximal proper  $*$-subalgebras of $\Md$. 
The maximum possible value of $\kappa$ is the longest chain of the following subalgebras
\begin{equation}\label{chain}
\mathcal{M}_{\Ep}\supseteq\mathcal{M}_{\Ep^2}\supseteq
\cdots \supseteq \mathcal{M}_{\Ep^{\kappa}}=\Mi.
 \end{equation}
Now note that the subalgebra $\Me$ can not be taken $\Md$ to begin with. This would mean the multiplicative domain of $\Ep$ is $\Md$ and hence $\Ep$ is a homomorphism and the kernel of $\Ep$,  $\rm{Ker(\Ep)}$ is an ideal of $\Md$. Since $\Md$ is simple as an algebra that is it can not contain a non-trivial two sided ideal, $\Ep$ has to have trivial kernel which makes $\Ep$ an isomorphism. So there exists another unital homomorphism $\Phi:\Md:\rightarrow\Md$ such that $\Ep\circ\Phi(x)=\Phi\circ\Ep(x)=x$, for all $x\in \Md$. Hence we get for every $x$,
\[\|x\|=\|\Phi\circ\Ep(x)\|\leq \|\Ep(x)\|\leq\|x\|.\]
This makes $\Ep$ an isometric isomorphism and hence following Kadison's work in \cite{kadison} we conclude that $\Ep(x)=uxu^*$ for some unitary $u$ in $\Md$, that is $\Ep$ is a unitary channel. Now by Example \ref{unitary} we get $\Mi=\Me=\Md$ and $\kappa=1$.

To get a possible value of $\kappa$ which is bigger than $1$, one starts with a channel $\Ep$ with $\Me$ to be a proper subalegbra of $\Md$. This choice significantly reduces the path of the the chain in Equation \ref{chain} because of the dimension criteria for any proper $*$- subalgebras. In a very recent article \cite{Agore}, it was shown that the maximal dimension of a proper unital subalgebra of $\Md$ is $d^2-d+1$. Since we are concerned with proper subalgebras which are also $*$-closed, the dimension of a maximal proper unital $*$-subalgebra of $\Md$ will be bounded by $d^2-d+1$. With this choice of $\Me$, the algebra $\mathcal{M}_{\Ep^2}$ will also have lesser dimension than that of $\Me$ and at the end even if $\Mi=\mathcal{M}_{\Ep^{\kappa}}=\mathbb{C}1$, we will get $\kappa<d^2$. 

In the case of $d=3$, it is more easily seen. By the Wedderburn's theorem, the $*$-subalgebras of $\mathcal{M}_{3}$ are described below upto isomorphism
\[\mathcal{M}_{2}\oplus\mathcal{M}_{1},\mathcal{M}_{1}\oplus\mathcal{M}_{1}\oplus\mathcal{M}
_{1}, (\mathcal{M}_{1}\otimes 1_{2})\oplus\mathcal{M}_{1}
,\mathcal{M}_{1}\otimes 1_{3},
\]
where $\mathcal{M}_{1}=\mathbb{C}1$ and $1_{2}, 1_{3}$ are identity matrices in $\mathcal{M}_{2}$ and $\mathcal{M}_{3}$ respectively. Example \ref{kappa three} shows the chain in \ref{chain} starting with $\mathcal{M}_{1}\oplus\mathcal{M}_{1}\oplus\mathcal{M}
_{1}$ and going down all the way to $\mathcal{M}_{1}\otimes 1_{3}$ yielding $\kappa=3$. 

%%%%%%%%%%%%%%%%%%%%%
We end this section with the following proposition which asserts that $\Mi$ is generated by partial isometries. Recall that an element $v\in \Md$ is a partial isometry if $vv^*=p$ and $v^*v=q$ where $p,q$'s are projections. Before we state our proposition, we state the following theorem 
\begin{theorem}{(\cite{ergodic})}
Let $\Phi:\Md\rightarrow\Md$ be a unital channel given by $\Phi(x)=\displaystyle\sum_{j=1}^{n} a_{j}xa_{j}^*$. Then an element $a\in \Md$ satisfies 
$\Phi(a)=\lambda a$ with $\lambda=1$ if and only if
\[a_{j}a=\lambda aa_{j}, \ \forall j=1, \cdots, n.\]
Moreover, the eigenspace of $\Phi$ corresponding to $\lambda$ is (linearly) spanned by partial isometries. 
\end{theorem}
Now part 3 of Theorem \ref{main theroem}, $\Mi$ is generated by the eigenvectors corresponding to the peripheral eigenvalues. Since the above theorem asserts that each peripheral eigenvector is a linear combination of partial isometries, we conclude that $\Mi$ is generated by partial isometries. We note down this fact as the following proposition.
\begin{proposition}
For a unital channel $\Ep:\Md\rightarrow\Md$, 
the algebra $\Mi$ is generated by partial isometries. 
\end{proposition}

%\begin{example}
%In general if $\Ep:\Md\rightarrow\Md$ is a unital channel 
%with Kraus decomposition $\Ep(x)=\sum_{j=1}^{n}a_{j}xa_{j}^*$ for some operators 
%$a_{j}\in \Md$ and for all $x\in \Md$, then we have a precise description for the eigen operators. In \cite{ergodic}, Theorem 11 asserts that an operator $a\in \Md$ satisfies $\Ep(a)=\lambda a$ for $|\lambda|=1$ if and only if $a_{j}a=\lambda aa_{j}$ for every $j$. Moreover, the eigenspace of $\Ep$ corresponding to $\lambda$ is spanned by partial isometries. Hence for any unital quantum channel $\Ep$, $\Mi$ is the algebra generated by partial isometries.    
%\end{example}

\section{Irreducibility and Primitivity }{\label{irreducible section}}
In this section we study the fixed points and multiplicative properties of irreducible positive linear maps on $\Md$. We recall some definitions and mention some well known facts below.
\begin{definition}
A positive linear map $\Phi:\Md \rightarrow \Md$ is called irreducible if there exist no non trivial projection $p \in \Md$ such that 
\begin{center}
$\Phi(p)\leq \lambda p$ for $\lambda>0$.
\end{center}
\end{definition}
We note some basic facts about irreducible positive linear maps on finite dimensional $C^*$ algebras (see also \cite{evans-krohn}, \cite{farenick}). In the literature of quantum information theory, such maps are also known as \emph{ergodic} linear maps (see \cite{ergodic}). In what follows ${\Md}_{+}$ denotes the set of all positive semidefinite elements of $\Md$.
We start with the spectral properties of an irreducible positive linear map.  
\begin{theorem}(see \cite{evans-krohn})
Let $\Phi$ be a positive linear map on $\Md$ and let $r$ be its spectral radius. Then
\begin{enumerate}
\item There is a non zero $x \in {\Md}_{+}$ such that $\Phi(x)=rx$
\item If $\Phi$ is irreducible and if $y \in {\Md}_{+}$ is an eigenvector of $\Phi$ corresponding to some
eigenvalue $s$ of $\Phi$, then $s = r$ and $y$ is a positive scalar multiple of $x$.\\
\item If $\Phi$ is unital, irreducible and satisfies the Schwarz inequality for positive linear maps then
\begin{itemize}
\item $r=1$ and $\mathcal{F}_{\Phi}=\mathbb{C}1$.
\item Every peripheral eigenvalue $\lambda \in \rm{Spec(\Phi)\cap \mathbb{T}} $ is simple and the corresponding eigenspace is spanned by a unitary $u_{\lambda}$ which satisfies $\Phi(u_{\lambda}x)=\lambda u_{\lambda}\Phi(x)$, for all $x\in \Md$.
\item The set $\Gamma=\rm{Spec(\Phi)\cap\mathbb{T}}$ is a cyclic subgroup of the group $\mathbb{T}$ and the corresponding eigenvectors form a cyclic group which is isomorphic to $\Gamma$ under the isomorphism $\lambda \rightarrow u_{\lambda}$. 
\end{itemize}
\end{enumerate}
\end{theorem}
With these set up we are ready to note down the following definition.
\begin{definition}
Let $\Ep:\Md\rightarrow\Md$ be a irreducible positive linear map. If $\Ep$ has trivial peripheral spectrum, then $\Ep$ is called primitive. 
 \end{definition}
 We refer to \cite{wolf} for some properties of primitive maps.
%In order to see that we first give this following theorem.
%We give note some basic prperties of a primitive map in the following theorem.
%\begin{theorem}
%Let $\mathcal{A}$ be a finite dimensional $C^*$ algebra and $\Ep:\mathcal{A}\rightarrow\mathcal{A}$ be a unital positive linear map. Then the following statements are equivalent
%\begin{itemize}
%\item $\Ep$ is primitive i.e $\Ep$ is irreducible and it's peripheral spectrum, $\Gamma=\mathbb{C}1$.
%\item $\Ep^n$ is irreducible for all $n \in \mathbb{N}$. 
%\end{itemize}
%\end{theorem}
%\begin{proof}
%to be written soon
%\end{proof}
We begin with describing the set $\Mi$ for a irreducible and primitive channel $\Ep$.
\begin{lemma}\label{abelian}
Let $\Ep$ be a unital irreducible channel. Then $\Mi$ is a commutative C$^*$-algebra. 
\end{lemma}
\begin{proof}
Since the channel $\Ep$ is unital and satisfies the Schwarz inequality, the peripheral spectrum $\Gamma=\rm{Spec(\Ep)\cap \mathbb{T}}$ is a cyclic subgroup. So $\Gamma=exp(2\pi i\mathbb{Z}_{m})$ for some $m\leq d^2$. If $u$ is the eigen vector for the eigenvalue $\lambda=exp(2\pi i/m)$, then it is easy to see that $\Ep(u^k)=\lambda^k u^k$ for all $k\in \mathbb{N}$. This shows that the peripheral eigen operators are generated by the powers of $u$ and hence $\Ne$ is spanned by the single unitary $u$. Since $\Mi$ is then algebraically generated by a unitary, we get $\Mi$ is commutative C$^*$-algebra.    
\end{proof}
Note that from the work of St{\o}rmer in \cite{stormer2008} it can be deduced that if $\Ep$ is unital \emph{entanglement breaking channel}, then $\Me$ (and hence $\Mi$) is abelian C$^*$-algebra where an entanglement breaking channel is one whose all Kraus operators are of rank one \cite{entng-brkng}. Lemma \ref{abelian} reflects upon similar characterisation of $\Mi$ in the case of irreducible channels which are found in abundance in the set of quantum channels.
\begin{corollary}
A unital quantum channel is primitive if and only if $\Mi=\mathbb{C}1$.
\end{corollary}
The following proposition captures the description of the stabilising algebra of a composition of two channels. The commutativity of channels is a necessary condition. 
\begin{proposition}{\label{important prop}}
If two unital channels $\Phi,\Psi$ commute that is $\Phi\circ\Psi=\Psi\circ\Phi$, then $\mathcal{M}_{({\Phi\circ\Psi})^\infty}=\{a\in \mathcal{M}_{{\Psi}^\infty}: \Psi^{k}(a)\in \mathcal{M}_{{\Phi}^\infty} \ \forall k\in \mathbb{N}\}$.
\end{proposition}
\begin{proof}
The proof is based on the same idea which was used in Lemma \ref{composition}. One side is straight forward. For the other side, let $a\in \mathcal{M}_{({\Phi\circ\Psi})^\infty}$. Then for any 
$k\in \mathbb{N}$, we have
\begin{align*}
 (\Phi\circ\Psi)^k(aa^*)&=\Phi^k(\Psi^k(aa^*))\\
 &\geq \Phi^k\{(\Psi^k(a))\Psi^k(a^*)\}\\
 &\geq\Phi^k(\Psi^k(a))
 \Phi^k(\Psi^k(a^*))\\
 &=\Phi^k\Psi^k(aa^*)\\
 &=(\Phi\circ\Psi)^k(aa^*).
\end{align*}
Since the two extreme ends are same, the inequalities become equality and using the trace preservation of $\Phi,\Psi$ and faithfulness of trace, we obtain the result.
\end{proof}

Note that starting from the Perron-Frobenius theory of positive and non-negative matrices, the notions of irreducibility and primitivity of matrices have attracted wide attention in matrix theory \cite{H-J}. These concepts have been applied in algebraic graph theory, 
finite Markov chains and other related fields. In 
\cite{Schwarz}, the author is concerned with irreducibility and primitivity of sums and products of
non-negative matrices. Since unital quantum channels can be thought of as non-commutative generalisations of bistochastic matrices (matrices with non-negative real entires where each row and column sums to 1), the remaining part of this section addresses similar notions of irreducibility and primitivity of products of unital channels.    

 Note that product of primitive channels need not be primitive. To get a glimpse of the importance of primitivity of product of channels, we observe that to  quantify the increment of entropy of a quantum system under a primitive channel, the logarithmic-Sobolev (LS) constant plays an important role \cite{M-S-D-W}. The discrete LS constant in \cite{M-S-D-W} is defined assuming $\Ep^*\circ\Ep$ is primitive. So primitivity of composition of two channels seems to be a useful assertion. We derive a necessary condition for the composition of two channels to be primitive.  
%\newpage
\begin{theorem}\label{composition of primitives}
Let two unital channels $\Phi,\Psi$ be such that $\Phi\circ\Psi=\Psi\circ\Phi$. If one of the two channels is primitive, then $\Phi\circ\Psi$ is primitive.
\end{theorem}
\begin{proof}
Proposition \ref{important prop} implies that   \[\mathcal{M}_{({\Phi\circ\Psi})^{\infty}}\subseteq
\mathcal{M}_{{\Phi}^{\infty}}\bigcap
\mathcal{M}_{{\Psi}^{\infty}}.\]
Now if one of them (say $\Phi$) is primitive, then $\mathcal{M}_{{\Phi}^\infty}=\mathbb{C}1$, forcing $\mathcal{M}_{({\Phi\circ\Psi})^{\infty}}=\mathbb{C}1$, making $\Phi\circ\Psi$ primitive.
\end{proof}
Now we give a new proof to the following theorem concerning the primitivity of $\Ep^*\circ\Ep$, which was first given in \cite{ergodic}, Theorem 13.
\begin{theorem}
Let $\Ep$ be a unital channel. If $\Ep^*\circ\Ep$ is irreducible, then $\Ep$ is primitive. Furthermore, if $\Ep^*\circ\Ep=\Ep\circ\Ep^*$ and $\Ep$ is primitive, then $\Ep^*\circ\Ep$ is primitive and hence irreducible.
\end{theorem}
\begin{proof}
Let $\Ep^*\circ\Ep$ be irreducible. Then we have $\mathcal{F}_{\Ep^*\circ\Ep}=\mathbb{C}1$. It is easy to see that if for any $a$, $\Ep(a)=\lambda a$ with $|\lambda|=1$, then $a\in \Me$. Now by Lemma \ref{primitive} we get $\Me=\mathcal{F}_{\Ep^*\circ\Ep}=\mathbb{C}1$. Hence the peripheral spectrum has to be trivial. So $\Ep$ is primitive.

The other assertion is a direct application of Theorem \ref{composition of primitives}. 

\end{proof}
The following result appeared in \cite{S-W}, Lemma 2,  which was used to show that the set of primitive channels with a fixed Kraus rank is path connected. We give an alternate proof of the result using the techniques we have developed. 
\begin{theorem} Let $\Ep$ be a unital quantum channel. Then 
 $\Ep$ is irreducible if and only if $\Ep+\Ep^2$ is primitive. 
\end{theorem}
\begin{proof}
First we observe that $\Ep+\Ep^2=\Ep(1+\Ep)=(1+\Ep)\Ep$ and hence $\Ep$ and $1+\Ep$ commute. Now it is well known (see \cite{wolf}, Theorem 6.2, 6.7) that if $\Ep$ is irreducible, then $1+\Ep$ is primitive. Now by the Proposition \ref{important prop},   $\mathcal{M}_{({\Ep+\Ep^2})^{\infty}}=
\mathcal{M}_{(\Ep(1+\Ep))^{\infty}} \subseteq\Mi\bigcap\mathcal{M}_{(1+\Ep)^\infty}$. As $\mathcal{M}_{(1+\Ep)^\infty}=\mathbb{C}1$, we have $\mathcal{M}_{({\Ep+\Ep^2})^{\infty}}$ is $\mathbb{C}1$, forcing $\Ep+\Ep^2$ to be primitive.

Conversely, let $\Ep+\Ep^2$ be primitive. If $\Ep$ is not irreducible, then there exists a non trivial projection $p$ such that $\Ep(p)\leq \lambda p$. Furthermore, using the Kraus operators it can be shown that $\Ep(p)\leq p$ (\cite{bratteli}, Lemma 3.1). Since $\Ep$ is trace preserving, we get by faithfulness of trace, $\Ep(p)=p$. This yields, $(\Ep+\Ep^2)(p)=2p$ and subsequently $(\Ep+\Ep^2)^{n}(p)=2^{n}p$. Now following \cite{wolf}, Theorem 6.7, if $\Ep+\Ep^2$ is primitive, then $\lim_{n\to\infty}(\Ep+\Ep^2)^n(p)$ is a positive definite matrix. But if $(\Ep+\Ep^2)^{n}=2^np$, then $\lim_{n\to\infty}(\Ep+\Ep^2)^n(p)$ does not exists, contradicting the primitivity of $(\Ep+\Ep^2)$.
\end{proof}

Exploiting the relation $\mathcal{F}_{\Ep^*\circ\Ep}=\Me$, we give a different proof of the following proposition given in \cite{ergodic}, Proposition 2.
\begin{proposition}
Let $\Ep$ be a unital channel. Then $\Ep^*\circ\Ep$ is irreducible if and only if there is no projection $p<1$ and no unitary $u$ such that $\Ep(p)=upu^*$.
\end{proposition}
\begin{proof}
Let $\Ep^*\circ\Ep$ be irreducible. Then  $\mathbb{C}1=\mathcal{F}_{\Ep^*\circ\Ep}=\Me$. If there is a projection $p$ and a unitary $u$ such that $\Ep(p)=upu^*$, then $\Ep(p)^2=up^2u^*=upu^*=\Ep(p)=\Ep(p^2)$ and hence $p\in \Me$, contradicting the hypothesis.

Conversely, assume the contrary that is  $\mathcal{F}_{\Ep^*\circ\Ep}\neq
\mathbb{C}$, then $\Me\neq\mathbb{C}1$ and hence assume $a\in \Me$. As $\Me$ is a C$^*$-algebra, replacing a by the real part $\Re a$ and the imaginary $\Im a$ part of $a$, we can assume $a$ is self adjoint. Furthermore, replacing $a$ by $a+\|a\|1$ and noting that $\Ep$ is unital, we can assume $a$ is positive. By the spectral decomposition theorem, let $a=\lambda_{1}p_{1}+\cdots\lambda_{k}p_{k}$ where $\lambda_{i}$'s are the eigenvalues and $p_{i}$'s are the corresponding eigen projections. Since $a\in \Me$, it is easy to see that $p_{i}\in \Me$, for every $i=1, 2, \cdots, k$. Hence for any such $j$, using Theorem \ref{projection and unitary}  we get $\Ep(p_{j})$ is a projection and because $\Ep$ is trace preserving, $\Ep(p_{j})$ is a projection of the same rank as $p_{j}$. This means there is a unitary $u$ such that $\Ep(p_j)=up_{j}u^*$. This violates the assumption.
\end{proof}
%\begin{theorem}
\section{In the Heisenberg Picture}{\label{heisenberg pic}}
In general setting, a quantum channel ($\Ep$) that is a completely positive trace preserving linear maps is defined on the trace class operators $\mathcal{T}(\mathcal{H})$ where $\mathcal{H}$ can be an infinite dimensional Hilbert space. The positivity of a map on $\mathcal{T}(\mathcal{H})$ is realised by considering $\mathcal{T}(\mathcal{H})$ as a matrix ordered space. 
 This map $\Ep:\mathcal{T}(\mathcal{H})\rightarrow \mathcal{T}(\mathcal{H})$ is seen as a linear operator in the Schr\"odinger picture. Since the dual of $\mathcal{T}(\mathcal{H})$ is $\mathcal{B}(\mathcal{H})$, that is $\mathcal{B}(\mathcal{H})_{*}=\mathcal{T}(\mathcal{H})$, the map  $\Ep$ induces a unital, normal and completely positive map on $\mathcal{B}(\mathcal{H})$. 
The dual picture where a completely positive and trace preserving map on trace class operators induces a unital completely positive map $ \Ep^*:\mathcal{B}(\mathcal{H})\rightarrow \mathcal{B}(\mathcal{H})$ is also an important association to a channel. 
The two maps are related via the relation $\tr(a\Ep(b))=\tr(\Ep^*(a)b)$ for all $a\in \mathcal{B}(\mathcal{H}), b\in \mathcal{T}(\mathcal{H})$.
This scenario where the map $\Ep^*:\mathcal{B}(\mathcal{H})\rightarrow\mathcal{B}(\mathcal{H})$ is a unital normal and completely positive, is known as Heisenberg picture. In finite dimensional Hilbert spaces, both the spaces are same but thought of as different matrix ordered spaces. That is, although the matrix ordered spaces ${\Md}_{*}$ and $\Md$ are different, a linear map, if it is completely positive on ${\Md}_*$ is equivalent to being completely positive on $\Md$. The discussion on positivity and quantum channel is mentioned in detail in \cite{farenick1}, Section III.A . We discuss the multiplicative property of a linear map in the Heisenberg picture and hence we will consider unital completely positive (ucp) maps on $\Md$.  
\paragraph{Density of Trivial Multiplicative Domain Maps}
\begin{theorem}\label{density of trivial mult. dom}
The set of unital completely positive maps on $\Md$ that have trivial multiplicative domains is dense in the completely bounded norm topology inside the set of unital completely positive maps.
\end{theorem}
\begin{proof}
Lets start with an arbitrary ucp map $\Phi$ and define a new map $\mathcal{E}:\Md\rightarrow\Md$
by $\Ep(x)=(1-\frac{1}{n})\Phi(x)+\frac{1}{n}\tr(x)\frac{1}{d}$, for all $x\in \Md$ where $n(>1) \in \mathbb{N}$ and $1$ is the identity operator in $\Md$.\\
One can check $\Ep$ is unital completely positive. Note also that $\Ep$ is strictly positive which means it sends positive semi definite operators to positive definite i.e invertible and positive elements. Hence it is clear that $\Ep$ is irreducible.
Now we see $\Ep$ approximates $\Phi$. Since the map $\tilde{\Ep}(x)=\tr(x)\frac{I}{d}$ is unital and positive and the cb norm is attained at $1$. We have 
\begin{align*}
||\Phi-\Ep||_{cb}&=||\frac{1}{n}(\Phi-\tilde{\Ep})||_{cb}\\
&\leq \frac{1}{n}[||\Phi||_{cb}+||\tilde{\Ep}||_{cb}]\\
&=\frac{2}{n}
\end{align*}
As we can take $n$ large enough, it shows that $\Ep$ approximates an arbitrary ucp maps $\Phi$.

Suppose $\Ep$ has a non trivial multiplicative domain and let $a \in \mathcal{M}_{\Ep}$. Since $\Me$ is a $C^*$ algebra, $a^*$ is also in $\mathcal{M}_{\Ep}$.
So we can assume $a$ to be self adjoint. Also since $\Ep$ is unital by replacing $a$ by $a+||a||1$ we can assume $a$ is a positive operator. Now let us assume $a$ has the following spectral decomposition:  $a=\lambda_{1}p_{1}+\cdots+\lambda_{k}p_{k}$ where each $p_{i}$ is a projection. Since 
$\mathcal{E}(a^2)=\Ep(a.a)=\Ep(a)\Ep(a)$ we have 
$\Ep(a^m)=\Ep(a)^m$ for all $m>0$. So for any polynomial $f$, we have $\Ep(f(a))=f(\Ep(a))$.
So for any polynomials $f$, $f(a)$ is in the multiplicative domain of $\Ep$ and since it is a $C^*$-algebra, all the spectral projections $p_{i}\in \Me$, for every $i$. By Theorem \ref{projection and unitary}  $\Ep(p_{i})$ is also a projection. Since $\Ep$ is strictly positive, the only possibility of $\Ep(p_{i})$ is $1$.  Now, if $\Ep(p_{i})=1$, then $\Ep(1-p_{i})=0$ which violates the irreducibility of $\Ep$. 
This shows the multiplicative domain of $\Ep$ can not contain any non-trivial elements.
\end{proof}
From the standpoint of C$^*$-algebra theory on finite dimensional Hilbert spaces, Theorem \ref{density of trivial mult. dom} reveals an interesting property that the majority of ucp maps on $\Md$ don't display any multiplicative nature even when restricted to a subdomain. From the viewpoint of quantum information theory, specially in \emph{quantum error correction} where a tacit relationship between multiplicative domain of a unital channel and error correcting codes have been discovered (see \cite{c-j-k}, \cite{johnston}), Theorem \ref{density of trivial mult. dom} rules out the possibility of \emph{perfect error correction} for majority of channels. Section \ref{sec:application} contains further discussion on 
the effect of multiplicative domain of a unital channel on error correction. 

We end this subsection with another observation which reinforce the effect of the Theorem \ref{density of trivial mult. dom} in the study of the convex set of all ucp maps on $\Md$. We first state the following Theorem
\begin{theorem}{\rm{(\text{Choi},\cite{choi1})}}
If $\Ep,\Phi,\Psi$ are ucp maps on a C$^*$-algebra $\mathcal{A}$ such that $\Ep=\frac{1}{2}(\Phi+\Psi)$, then
\[\Me=\mathcal{M}_{\Phi}\cap\mathcal{M}_{\Psi}\cap
\{x\in \mathcal{A}:\Ep(x)=\Phi(x)=\Psi(x)\}.\]
\end{theorem}    
Now if a ucp map $\Phi$ is such that $\mathcal{M}_{\Phi}=\mathbb{C}1$, then Choi's theorem 
implies that any ucp map lying in the line-segment passing through $\Phi$ must have trivial multiplicative domain. Since by Theorem \ref{density of trivial mult. dom} such $\Phi$'s are dense, it follows that every line-segment passing through the elements of this dense set contains ucp maps with trivial multiplicative domain. Hence ucp maps which do not show any multiplicative nature is quite generic in general sense.  
%%%%%%%%%%%%%%%%%%%
\paragraph{On Arveson's Boundary type theorem on matrix algebras}
Recall that an operator $a$ on a Hilbert space $\mathcal{H}$ is called \emph{irreducible} if there exists no non trivial projection $p$ such that $ap=pa$. The celebrated Boundary theorem (\cite{arveson}) of Arveson in finite dimension asserts:
\begin{theorem}{\label{Boundary Theorem}}(\cite{farenick2})
If $a\in \Md$ is irreducible and if $\Phi:\Md \rightarrow \Md$ is a unital completely positive linear transformation such that $a \in \mathcal{F}_{\Phi}$, then $\Phi(x)=x$ for all $x \in \Md$.
\end{theorem}
We state a simple lemma, the proof of which can be easily derived using the fact that the algebra generated by $\{1, a, a^*\}$ is $\Md$, where $a$ is an irreducible operator.
\begin{lemma}
Let $\Ep:\Md \rightarrow \Md$ be a ucp map and let $a\in \Md$ be irreducible such that $a \in \mathcal{M}_{\Ep}$.
Then $\Ep$ is an automorphism.
\end{lemma}
\begin{proof}
Easy.
\end{proof}
Now we state and prove the main theorem of the section. Since in the discussion about multiplicative index, the eigen operators play an important role, the following theorem provides some more information extending the Theorem \ref{Boundary Theorem}.
 For a ucp map $\Ep:\Md\rightarrow \Md$, 
we recall that $\Ne=\{x\in \Md:\Ep(x)=\lambda x, |\lambda|=1\}$. Clearly the set $\Ne$ contains the fixed point set $\mathcal{F}_{\Ep}$. Then we have the following theorem.
\begin{theorem}{\label{pripheral eigenvectors}}
Let $\Ep:\Md\rightarrow\Md$ be a ucp map. Suppose an irreducible operator $a\in \Md$ is inside $\Ne$. Then
$\Ne$, the peripheral eigen operators of $\Ep$, span the entire $\Md$. 
\end{theorem} 
\begin{proof}

For a ucp $\Ep$, Kuperberg in \cite{kup} proved that there is a sequence of integers $n_{1}<n_{2}< \cdots$ such that $\Ep^{n_{k}}$ converges to a unique 
idempotent and completely positive map $P:\Md \rightarrow \Md$ and $Ran(P)$=span-$\Ne$.
Now we proceed similarly as in \cite{farenick2} and  using Choi-Effrose product on span-$\Ne$ we can make 
it a C$^*$-algebra. Let us define $a\odot b=P(ab)$ where $a,b\in$ span-$\Ne$ and it is well known that with this multiplication span-$\Ne$=$Ran(P)$ is a C$^*$-algebra. Since $P$ is a conditional expectation, we have $P(yz)=P(yP(z))$ for $y\in$ span-$\Ne, \  z\in \Md $. 
Now, we see for any 
$a \in$ span-$\Ne$ , $P(a)=a$ and\\
 $P(xy)=P(xP(y))=x\odot P(y)=P(x)\odot P(y)$, $x\in$ span- $\Ne, \ y\in \Md$.\\
 Similarly if $x_{1}, x_{2}\in$ span-$\Ne, \ y\in \Md$ we get\\
 $P((x_{1}x_{2})y)=P(x_{1}P(x_{2}y))=x_{1}\odot P(x_{2}y)=(P(x_{1})\odot P(x_{2}))\odot P(y)$.\\
 Hence by induction, if $\gamma$ is any word with $2n$ non commutative variables and if $x_{1}, \cdots,x_{n}\in$ span-$\Ne$ and $y \in \Md$ we have 
 \begin{align*}
 P(\gamma(x_{1}, \cdots, &x_{n},x_{1}^*, \cdots, x_{n}^*)y)\\
 &=\gamma_{\odot}(P(x_{1}), \cdots, P(x_{n}),P(x_{1}^*), \cdots, P(x_{n}^*))\odot P(y).
\end{align*}
Where $\gamma_{\odot}(P(x_{1}), \cdots, P(x_{n}),P(x_{1}^*), \cdots, P(x_{n}^*))$ denotes the $\odot$ product of the letters of the word $\gamma$.\\
Now since $\Ne$ contains an irreducible operator
$a$ and since $\Ne$ is $*$ closed we have $a^*\in \Ne$ . By irreducibility of $a$, the set $\{1,a,a^*\}$ algebraically generates the entire $\Md$. So for any $x\in \Md$ is of the form $\gamma(x_{1}, \cdots, x_{n},x_{1}^*, \cdots, x_{n}^*)$ with n varies over $\mathbb{N}$ and $x_{i}\in \Ne$ we have
\begin{center}
$P(xz)=P(x)\odot P(z)$ for all $x, z\in \Md$.
\end{center} 
So $P:\Md\rightarrow$ span-$(\Ne,\odot)$ is a homomorphism 
between the algebra $\Md$ and span-$(\Ne,\odot)$.\\
Since $\Md$ is simple as an algebra and $\Ne$ is not trivial, we get $Ker(P)$ is trivial, and hence $P$
is an isomorphism and therefore the set span-$\Ne$ is exactly the set $\Md$.  
\end{proof}
Note that in \cite{arveson2},\cite{arveson3} Arveson studied asymptotic stability of ucp maps on C$^*$-algebras where the maps in discussion satisfy certain properties. Isometric invertible operators acting on a general Banach space $X$  whose peripheral eigenvectors span the domain $X$ seem to have played a significant role in studying the asymptotic stability of ucp maps. Such maps were named \emph{diagonalisable} in 
\cite{arveson2}.  Theorem \ref{pripheral eigenvectors} exhibits similar behaviour of a ucp map if it contains an irreducible operator as a peripheral eigenvector. 
Since $\mathcal{F}_{\Ep}\subseteq\mathcal{N}_{\Ep}$ for a unital linear map, Theorem \ref{pripheral eigenvectors} also suggests a different prospective of the \emph{Boundary Theorem} of Arveson (Theorem \ref{Boundary Theorem}) on finite dimensional spaces from a more spectral theoretic viewpoint.

\section{Application}\label{sec:application}
In this section we note down some applications of the techniques and results developed throughout this paper. Quantum error correction is one of the areas where we find appropriate set up for exhibiting such applications. See \cite{knill-lafl},\cite{c-j-k} for more detailed discussion in the theory of error correction.

The standard model of error correction can be 
described by a triple $(\Ep,\mathcal{R},\mathcal{C})$. Here $\Ep:\Bh\rightarrow\Bh$ is a quantum channel, where $\Hi$ is a finite dimensional Hilbert space. $\mathcal{C}$ is a subspace of $\Hi$ known as the \emph{code} and $\mathcal{R}$ is another  quantum channel on $\Bh$ known as \emph{recovery operation}. Denote $P_{\mathcal{C}}$ 
by the projection onto $\mathcal{C}$. The triple should satisfy the condition 
\[\mathcal{R}(\mathcal{E}(\rho))=\rho,
 \ \text{where} \ P_{\mathcal{C}}\rho P_{\mathcal{C}}=\rho.\]
 With this set up, the ``Noiseless Subsystem" (NS) protocol (see \cite{NS1},\cite{knill-lafl}) seeks subsystems $\mathcal{H}^{B}$ (with dim $\mathcal{H}^B> 1$) of the full system $\Hi$ such that  $\mathcal{H}=(\mathcal{H}^A\bigotimes\mathcal{H}^B)
\bigoplus \mathcal{K}$, where $\mathcal{K}$ is another subspace of $\Hi$ such that $\forall \rho^A, \ \forall \rho^B$ there exists a $\gamma^A$ satisfying 

\[
\mathcal{E}(\rho^A\otimes \rho^B)=\gamma^A\otimes \rho^B.\]
Here we write $\rho^A$ (resp. $\rho^B$) for operators on
$\mathcal{B}(\mathcal{H}^A)$ (resp. $\mathcal{B}(\mathcal{H}^B)$). Noiseless subsystems of any unital channel $\Ep$ are obtained precisely from the fixed point set $\mathcal{F}_{\Ep}$ of $\Ep$.  

A subsystem $\mathcal{B}$ is called \emph{correctable} for $\Ep$ via a recovery operation $\mathcal{R}$ if it is a noiseless subsystem for the quantum operation $\mathcal{R}\circ\Ep$. An important class of subsystems is the \emph{unitarily correctable subsystem} (UCS) which is also known as \emph{unitarily correctable codes} where the recovery operation $\mathcal{R}$ can be chosen to be the unitary channel $x\mapsto uxu^*$, for a unitary operator $u$ and for all $x\in \Bh$. The following theorem relates the UCS of a unital channel $\Ep$ to the noiseless subsystems of $\Ep^*\circ\Ep$.
\begin{theorem}{\cite{kribs-spekkens}}\label{ucc}
Let $\Ep$ be a unital quantum channel. Then the following statements are equivalent:
\begin{enumerate}
\item $\mathcal{B}$ is a unitarily correctable subsystem for $\Ep$.
\item $\mathcal{B}$ is a noiseless subsystem for $\Ep^*\circ\Ep$.
\end{enumerate}
\end{theorem}  
In \cite{c-j-k}, the UCC algebra for a unital channel is defined to be the $\mathcal{F}_{\Ep^*\circ\Ep}$ and in Theorem 11 it was shown that this is precisely the multiplicative domain.

With these background we are ready to exhibit the application of the techniques developed in this paper. We will prove that for a unital channel $\Ep$, even if we require the recovery operation $\mathcal{R}$ to be a unital channel (not necessarily a unitary channel), we still get the multiplicative domain of $\Ep$ to be the exact correctable code. This means we don't get any extra correctable codes  other than those arising from the multiplicative domain as in Theorem \ref{ucc} even if we allow our recovery map to be any unital channel. We formulate the following proposition:
\begin{proposition}\label{appl. 1}
Given a unital channel $\Ep:\Md\rightarrow\Md$, define the following set
\[\mathfrak{C}_{\Ep}=\{\mathcal{F}_{\mathcal{R}\circ\Ep}: \text{for unital channels} \ \mathcal{R} \ on \ \Md \}.\]
Then we have   
\[ \Me= \mathfrak{C}_{\Ep}.\] 
\end{proposition}   
\begin{proof}
First we see that $\Me\subseteq\mathfrak{C}_{\Ep}$ by the aid of the Theorem \ref{mult.dom thm} that asserts $\Me=\mathcal{F}_{\Ep^*\circ\Ep}$. 
Hence $\Ep^*$ is one of the choices of $\mathcal{R}$.

Conversely let $a\in \Md$ be such that $a\in \mathcal{F}_{\mathcal{R}\circ\Ep}$ for a unital channel $\mathcal{R}$. Since $a^*\in \mathcal{F}_{\mathcal{R}\circ\Ep}$ as well, we have 
\begin{align*}
\tr(aa^*)&=\tr(\mathcal{R}\circ\Ep(aa^*))\\
&\geq \tr(\mathcal{R}(\Ep(a)\Ep(a^*)))\\
&\geq \tr(\mathcal{R}\circ\Ep(a)\mathcal{R}\circ\Ep(a^*))\\
&=\tr(aa^*).
\end{align*}
Where we have use the Schwarz inequality for the unital maps $\Ep$ and $\mathcal{R}$. So the inequalities above are all equalities. Subsequently, by the trace preservation property of $\mathcal{R}$, we get $\tr(\Ep(aa^*))=\tr(\Ep(a)\Ep(a^*))$ and hence by the faithfulness of trace and the Schwarz inequality for $\Ep$ we get $\Ep(aa^*)=\Ep(a)\Ep(a^*)$, which shows $a\in \Me$. Hence we have the desired equality of sets.
\end{proof}
Since the unitary channels are a particular case for  arbitrary unital channels, it can be noted that Theorem 
\ref{ucc} appears as a special case of the Proposition 
\ref{appl. 1}.
It is now evident from Proposition \ref{appl. 1} that  even if we collect the \emph{unitally correctable codes} of a unital channel via a unital recovery operation, we don't get anything other than the \emph{unitarily correctable codes}.

A UCS code $\mathcal{B}$ of $\Phi$ is called \emph{unitarily noiseless subsystem} (UNS) for $\Phi$ (see \cite{Magesan}, \cite{UNS} ) if $\mathcal{B}$ is UCS code for $\Phi^n$ for every $n\geq 1$. Now by the aid of the Theorem \ref{ucc}, the UCS codes of $\Phi^n$ are exactly the noiseless subsystems for $\Phi^{*n}\circ\Phi^n$ for $n\geq 1$, which can be obtained precisely from the algebra $\mathcal{F}_{\Phi^{*n}\circ\Phi^n}$. Since for each $n\in \mathbb{N}$, $\Phi^n$ is a unital channel, Theorem \ref{mult.dom thm} asserts that
\[\mathcal{M}_{\Phi^n}=\mathcal{F}_{\Phi^{*n}\circ\
\Phi^n}.\]
So if $\mathcal{B}$ is a UNS code, then it can be obtained from the set $\displaystyle\bigcap_{n\geq 1}\mathcal{F}_{\Phi^{*n}\circ\Phi^n}$ which is essentially the stabilising algebra $\mathcal{M}_{\Phi^{\infty}}$. Hence the UNS codes for a unital channel $\Phi$ are exactly those that arise from the stabilising subalgebra $\mathcal{M}_{\Phi^{\infty}}$. 

In this connection the following proposition relates to the notion of multiplicative index of a unital channel and the UNS codes.
\begin{proposition}\label{UNS}
Let $\Ep:\Md\rightarrow\Md$ be a unital channel with multiplicative index 1, that is $\kappa(\Ep)=1$. Then every UCS code for $\Ep$ is a UNS code. Moreover, if $\kappa(\Ep)>1$, then there exist UCS codes which are not UNS.
\end{proposition}
\begin{proof}
As $\kappa(\Ep)=1$, $\Mi=\mathcal{M}_{\Ep^n}=\Me$, for every $n\geq 1$. Hence the stabilising subalgebra 
$\Mi$ is obtained after applying the channel only once.  
Clearly if $\mathcal{B}$ is UCS then $\mathcal{B}$ is obtained from the set $\Me=\Mi$ which is same as $\mathcal{M}_{\Ep^n}$, for all $n\geq 1$ and the first assertion follows.

For the other assertion, note that if $\kappa(\Ep)>1$, then we have proper containment  $\Mi\subset\Me$. Hence there is a UCS code  $\mathcal{B}$ corresponding to $\Me$ such that it is not UCS for $\Ep^2$, and hence not UCS for $\Ep^n$ if $n\geq 2$. So $\mathcal{B}$ ca not be UNS.
\end{proof}
Note that by the Corollary \ref{diagonalisable}, all the channels that commute with the adjoints have multiplicative index 1 and hence have UCS codes as UNS codes.  
As a corollary to the above result we capture the following well known result for Pauli Channels which constitute an important class of quantum channels . 
\begin{corollary}{(\cite{Magesan1}, Section 3.2.1)}
A UCS code for Pauli Channels or generalised Pauli channels is also a UNS code.
\end{corollary} 
\begin{proof}
If $\Ep$ is a Pauli channel or a generalised Pauli channel, then $\Ep^*\circ\Ep=\Ep\circ\Ep^*$ (\cite{ergodic}, see discussion after definition 6). Hence by the Corollary \ref{diagonalisable}, we get 
$\Me=\Mi$ and hence $\kappa(\Ep)=1$ and Proposition \ref{UNS} applies.    
\end{proof}
\section{Summary and Discussion}\label{sec:summary}
We have put forward a structure theorem for a unital quantum channel that depicts the asymptotic automorphic behaviour of the channel on a stabilising subalgebra. This subalgebra is generated by the peripheral eigen operators of the channel. Based on the finite dimensionality of the system, the notion of multiplicative index has been introduced which measures the number of times the channel needs to be composed with itself for the multiplicative domain to stabilise. Some applications of the results obtained in the paper have been shown in quantum error correction.  

%%%%%%%%%%%%%%%%%%%%%%%%%%%
It is interesting to note that any unital channel $\Phi:\Md\rightarrow\Md$ can be approximated by a channel $\Ep$ with $\kappa(\Ep)=1$. Indeed the set  
 \[\mathfrak{S}_{1}=\{\Ep:\Md\rightarrow\Md: \kappa(\Ep)=1\}\] is dense inside the convex set of unital channels. For given a channel $\Phi$, proceeding similarly as in the proof of the Theorem \ref{density of trivial mult. dom}, for each $n\in \mathbb{N}$, there exists a channel  
 $\Ep$ such that $\|\Phi-\Ep\|_{cb}<\frac{1}{n}$ and 
 $\Me=\mathbb{C}1$. Each of these $\Ep$'s has  $\kappa(\Ep)=1$ as $\Mi=\Me=\mathbb{C}1$.
   
Even though $\mathfrak{S}_{1}$ is dense, it follows that $\mathfrak{S}_{1}$ is not relatively open set inside the set of unital channels. To see this, define a channel $\Phi:\mathcal{M}_{2}\rightarrow\mathcal{M}_{2}$ as follows 
\[\Phi(x)=pxp+qxq, \ \forall x\in \mathcal{M}_{2},\]
where $p=\left[\begin{array}{cc}
1 & 0\\
0 & 0
\end{array}\right], q=\left[\begin{array}{cc}
0 & 0\\
0 & 1
\end{array}\right]$.

Then it is easily verified that 
\[\mathcal{M}_{\Phi}=\Bigg\{\left[\begin{array}{cc}
a & 0\\
0 & b
\end{array}\right]: a,b\in \mathbb{C}\Bigg\}.\]
Also since $\Phi^*=\Phi$, we get from Corollary \ref{diagonalisable} that $\mathcal{M}_{\Phi}=\mathcal{M}_{{\Phi}^\infty}$ and hence $\kappa=1$ and we have $\Phi\in \mathfrak{S}_{1}$.

Now for each $t\in [0,1]$, define $\Phi_{t}:\mathcal{M}_{2}\rightarrow\mathcal{M}_{2}$ by 
\[\Phi_{t}(x)=p(t)xp(t)^*+q(t)xq(t)^*,\]
where $p(t)=\frac{1}{c}\left[ \begin{array}{cc}
1+t & 0\\
t & 0
\end{array}\right], q(t)=\frac{1}{c}\left[ \begin{array}{cc}
0 & -t\\
0 & 1+t
\end{array}\right], \text{and} \ c= \sqrt{(1+t)^2+t^2}$.

Then $\Phi_{t}$ is a unital quantum channel for each $t\in [0,1]$ such that $\Phi_{0}=\Phi$. So it follows that $t\mapsto\Phi_{t}$ is a continuous path starting from $\Phi$ and ending at $\Phi_{1}$. We will show that for each $t>0$, we will have $\kappa(\Phi_{t})>1$. Since any open neighbourhood of $\Phi$ will intersect this path, it will be evident then that $\Phi$ can not admit a neighbourhood consisting of channels with $\kappa=1$ only, implying the fact that $\mathfrak{S}_{1}$ is not a relatively open set. To this end, we note that for any $x\in \mathcal{M}_{2}$,
\[\Phi_{t}^*\circ\Phi_{t}(x)=pxp+qxq,\]
where $p= \left[\begin{array}{cc}
1 & 0\\
0 & 0
\end{array}\right]$ and $q=\left[\begin{array}{cc}
0 & 0\\
0 & 1
\end{array}\right]$.
Hence we have  
\[\mathcal{M}_{\Phi_{t}}=\Bigg\{\left[\begin{array}{cc}
a & 0\\
0 & b
\end{array}\right]: a,b\in \mathbb{C}\Bigg\}.\]
Now for $p= \left[\begin{array}{cc}
1 & 0\\
0 & 0
\end{array}\right]
\in \mathcal{M}_{\Phi_{t}}$, we compute $\Phi_{t}(p)=\frac{1}{c^2}\left[\begin{array}{cc}
(1+t)^2 & (1+t)t\\
(1+t)t & t^2
\end{array}\right]$. Evidently $\Phi_{t}(p)\not\in \mathcal{M}_{\Phi_{t}}$ if $t\neq 0$. And by definition we get $p\not\in \mathcal{M}_{\Phi_{t}^2}$. Hence $\mathcal{M}_{\Phi_{t}^2}\subset\mathcal{M}_{\Phi_{t}}
$ and subsequently $\kappa>1$ as desired.

It can be also seen that for any value $n>1$, the set \[
\mathfrak{S}_{n}=\{\Ep:\Md\rightarrow\Md: \kappa(\Ep)=n\}\]
can not be relatively open as $\mathfrak{S}_{1}$ is  a dense set. 

Notwithstanding the discussion on channels with fixed multiplicative index, it will be interesting to know:  
given a unital channel $\Ep$ on $\Md$, how to get the exact value of the multiplicative index ($\kappa$) of $\Ep$. 
The exact upper bound for $\kappa$ is an interesting problem itself. The discussion after the Example \ref{kappa three} suggests that it is strictly less than $d^2$. The question may have a connection with the exact dimension of proper maximal unital $*$-subalgebras of $\Md$. 
    
\section{Acknowledgements}
The work is supported by the Graduate Research Fellowship at the University of Regina.
The author likes to thank Dr.\ Douglas Farenick for various stimulating discussions, Dr.\ Sarah Plosker for reading the manuscript and providing with feedback and Sam Jaques for a lot of insightful discussions on this topic. The author also wants to thank the anonymous referees for various comments and suggestions to improve the quality and exposition of this work.

\bibliography{multiplicative}
\bibliographystyle{amsplain}

\end{document}